\documentclass[11pt]{article}

\input{preamble.sty}

\title{Variational Inference for Stochastic Block Models from Sampled Data}
\ifblinded
  \author{\textbf{Anonymous authors}}
  \date{}
\else
  \author{Timoth\'ee   Tabouy,  Pierre  Barbillon and Julien Chiquet}
  \date{UMR MIA-Paris, AgroParisTech,  INRA, Universit\'e Paris-Saclay, 75005 Paris, France}
\fi

\begin{document}

\maketitle

\begin{abstract}
  This paper deals with non-observed dyads during the sampling of a
  network and consecutive issues in the inference of the Stochastic
  Block Model (SBM).  We review sampling designs and recover Missing
  At Random (MAR) and Not Missing At Random (NMAR) conditions for the
  SBM.  We introduce variants of the variational EM algorithm for
  inferring the SBM under various sampling designs (MAR and NMAR) all
  available as an \texttt{R} package.  Model selection criteria based
  on Integrated Classification Likelihood are derived for selecting
  both the number of blocks and the sampling design.  We investigate
  the accuracy and the range of applicability of these algorithms with
  simulations.  We explore two real-world networks from ethnology
  (seed circulation network) and biology (protein-protein interaction
  network), where the interpretations considerably depends on the
  sampling designs considered.
  \\

  \noindent Stochastic Block Model $\cdot$ Variational inference $\cdot$ Missing data $\cdot$ Sampled network
\end{abstract}

\section{Introduction}

Networks arise in many fields of application for providing an
intuitive way to represent interactions between entities.  In this
paper, a network is composed by a fixed set of nodes, and an
interaction between a pair of nodes (dyad) is called an edge. We consider
undirected binary networks with no loop, which can be
represented by symmetric adjacency matrices filled with zeros and
ones. 

Various statistical models exist for depicting the probability
distribution of the adjacency matrix \citep[see, e.g.][for a
survey]{goldenberg2010survey,snijders2011statistical}.  A highly
desirable feature is their capability to describe the heterogeneity of
real-world networks. In this perspective, the family of models endowed
with a latent structure \citep[reviewed in][]{matias2014modeling}
offers a natural way to introduce heterogeneity.  Within this family
the Stochastic Block Model \citep[in short SBM,
see][]{frank1982cluster,holland1983stochastic} describes a broad
variety of network topologies by positing a latent structure (or a
clustering) on the nodes, then making the probability distribution of
the adjacency matrix dependent on this latent structure. In order to
estimate SBMs, Bayesian approaches were first developed
\citep{snijders1997estimation,Nowicki2001} prior to variational
approaches \citep{daudin2008mixture,latouche2012variational}. On the
theoretical side, \citet{celisse2012consistency} study the conditions
for identifiability and the consistency of the variational estimators;
\cite{bickel2013asymptotic} prove their asymptotic normality.  Several
generalizations are possible such as weighted or directed variants
\citep{mariadassou2010}, mixed-membership and overlapping SBM
\citep{airoldi2008mixed,latouche2011overlapping}, degree-corrected SBM
\citep{Karrer2011}, dynamic SBM \citep{matias2016statistical}, or
multiplex SBM \citep{Barbillon2015}.

This paper deals with inference in the SBM when the network is not
fully observed. We consider cases where all the nodes are observed but
information regarding the presence/absence of an edge is missing for
some dyads.  In other words the adjacency matrix contains missing
values, a situation often met with real-world networks.  For instance
in social sciences, network data consists in interactions between
individuals: the set of individuals is fixed, possibly known from a
census. Information about the presence/absence of an edge is only
available when at least one of the two individuals is available for an
interview, otherwise it is missing.  See \cite{thompson2000model},
\cite{thompsonseber}, \cite{Kolaczyk2009,Handcock2010} for a review of
network sampling techniques. Even though some papers deal with SBM
inference under missing data condition \citep{Aicher2014,Vinayak2014},
the sampling mechanism responsible for the missing values is
overlooked in the inference, contrary to the approach developed in our
paper.

\paragraph*{Our contributions.} A typology of sampling designs is
introduced in Section~\ref{subsec:missingdata}. We adapt the theory
developed in \cite{Rubin1976,little2014statistical} to the SBM by
splitting the sampling designs into the three usual classes of missing
data:
\begin{enumerate}[i)]
\item \textit{Missing Completely At Random} (MCAR), where the sampling
  does not depend on the data, neither on the observed nor on the
  unobserved part of the network.
\item \textit{Missing At Random} (MAR), where the probability of being
  sampled is independent on the value of the missing data.  For
  network data, the sampling does not depend on the presence/absence
  of an edge of an unobserved (or missing) dyad. MCAR is a particular
  case of MAR.
\item  \textit{Not  Missing  At  Random} (NMAR),  where  the  sampling
  scheme is  guided by unobserved  dyads in some way.
\end{enumerate}
Section~\ref{sec:sampling-definition} introduces several examples of
sampling designs (MAR and NMAR) for which we derive conditions for
identifiability of the SBM parameters.

Estimation of the SBM in the MAR cases can be handled with the
Variational EM (VEM) of \cite{daudin2008mixture} by conducting the
inference only on the observed part of the network
(Section~\ref{sec:MARinference}). NMAR is more difficult to deal with
as the sampling design must be taken into account in the inference.
We introduce in Section~\ref{sec:NMARinference} a general variational
algorithm \citep{jordan1998introduction} to deal with NMAR cases when
the sampling design relies on a probability distribution which is
explicitly known\footnote{More complex sampling schemes -- for
  instance adversarial strategies -- are thus not handled}.  Our
variational approach is based on a double mean-field approximation
applied to the latent distribution of the clustering and to the
distribution of the missing dyads. We implement VEM algorithms that
produce unbiased estimators for three natural NMAR sampling designs: a
dyad-centered strategy, a node-centered strategy, and a block-centered
strategy.  We also derive an Integrated Classification Likelihood
criterion \citep[ICL,][]{Biernacki2000} for selecting the number of
blocks.  Although it is not possible to distinguish whether the
sampling is MAR or NMAR \citep{Molenberghs2008}, the ICL can also be
used to select which sampling design is the best fit for the data.

In Section~\ref{sec:NMARsimulations} we show the good performance of
our VEM algorithms on simulations for both MAR
(Section~\ref{sec:MARsimulations}) and NMAR conditions. Finally we
investigate two very different real-world networks with missing
values, namely a Kenyan seed exchange network
(Section~\ref{sec:kenya}), and a protein-protein interaction (PPI)
network (Section~\ref{sec:er1}).

\paragraph*{Related works.}  In the few papers dealing with missing
data for networks, the sampling design is rarely discussed. Even if
not explicitly stated they all assume MAR
conditions. \citet{Aicher2014} propose a weighted SBM modeling
simultaneously the presence/absence of an edge and its weight. Missing
data are handled by dropping the corresponding terms in the likelihood
and the inference is conducted by a variational algorithm. In
\citet{VincentKyleandThompson2015} a Bayesian augmentation procedure
is introduced to estimate simultaneously the size of the population
and the clustering when the sampling design is a one-wave
snowball. Apart from the SBM, the exponential random graph model has
been studied in the MAR setting in \citet{Handcock2010}.

The matrix completion literature brings additional insights since SBM
inference can be seen as a low-rank matrix
estimation. \citet{Vinayak2014} introduce a convex program for the
matrix completion problem where the underlying matrix has a simple
affiliation structure defined via an SBM. The entries are sampled
independently with the same probability, corresponding to a MAR case.
In \citet{Davenport2014} the case of noisy 1-bit observations is
studied and a likelihood-based strategy is developed with theoretical
justifications ensuring good matrix completion. \citet{Chatterjee2015}
proves strong results for large matrices with noisy entries
estimation, by means of a universal singular value thresholding.

Another related question is when the status of some dyads
(absence/presence) is not clear in errorfully observed graph.  Such
uncertainties can be taken into account
\citep{Priebe2015,Balachandran2017}. The latter reference studies the
error propagation made by using estimators computed on observed
sub-graphs, in order to estimate the number of existing edges in the
real underlying graph.


\section{Statistical framework}

\subsection{Stochastic Block Model}

\label{sec:SBM}

In an SBM, nodes from a set $\node \triangleq \{1,\dots,n\}$ are
distributed among a set $\block \triangleq \{1, \dots, Q\}$ of hidden
blocks that model the latent structure of the graph. The blocks are
described by the latent random vectors
$\big(Z_{i\sbullet}=(Z_{i1},\ldots,Z_{iQ}) \big)_{i\in\node}$ with
multinomial distribution
$\mathcal{M}(1,\alpha = ( \alpha_{1}, \dots, \alpha_{Q}))$.  The
probability of an edge between any dyad in
$\dyad \triangleq \node \times \node$ only depends on the blocks the
two nodes belong to.  Hence, the presence of an edge between $i$ and
$j$, indicated by the binary variable $Y_{ij}$, is independent on the
other edges conditionally on the latent blocks:
\begin{equation}
\nonumber
\MA_{ij}    \   |    \   Z_{iq} = 1,    Z_{j\ell} = 1   \sim^{\text{ind}}
\mathcal{B}(\pi_{q\ell}), \qquad \forall (i,j) \in\dyad, \quad \forall
(q,\ell) \in\block\times\block,
\end{equation}
where $\mathcal{B}$ stands for the Bernoulli distribution.  In the
following,
$\pi = \left(\pi_{q\ell}\right)_{(q,\ell) \in\block\times\block}$ is
the $Q \times Q$ matrix of connectivity probabilities,
$Y=(Y_{ij})_{(i,j)\in\dyad}$ is the $n\times n$ adjacency matrix of
the random graph, $Z = (Z_{iq})_{i\in\node,q\in\block}$ is the
$n \times Q$ matrix of the latent blocks and $\theta=(\alpha, \pi)$
are the unknown parameters.  In the undirected binary case,
$Y_{ij} = Y_{ji}$ for all $(i,j) \in \dyad$ and $Y_{ii} = 0$ for all
$i\in\node$.  Similarly, $\pi_{q\ell}=\pi_{\ell q}$ for all
$(q,\ell)\in\mathcal{Q}\times \mathcal{Q}$.

\subsection{Sampled data in the SBM framework}

\label{subsec:missingdata}

The sampled data is an $n\times n$ matrix with entries in
$\{0, 1, \texttt{NA}\}$.  It corresponds to the adjacency matrix $Y$
where unobserved dyads have been replaced by \texttt{NA}'s.  More
formally, let $R$ be the $n\times n$ sampling matrix recording the
data sampled during this process, such that $R_{ij} = 1$ if $Y_{ij}$
is observed and $0$ otherwise; also define
$\mathcal{D}^o = \{(i,j) : R_{ij}=1\}$,
$\mathcal{D}^m = \{(i,j) : R_{ij}=0\}$,
$\MAO = \{\MA_{ij} : (i,j) \in \mathcal{D}^o\}$ and
$\MAM=\{\MA_{ij} : (i,j) \in \mathcal{D}^m \}$ to denote the sets of
variables respectively associated with the \textit{observed} and
\textit{missing} data. The number of nodes $n$ is assumed to be known.
The \textit{sampling design} is the description of the stochastic
process that generates $R$. It is assumed that the network exists
before the sampling design acts upon it.  Moreover, the sampling
design is fully characterized by the conditional distribution
$p_\psi(R|\MA)$, the parameters of which are such that $\psi$ and
$\theta$ live in a product space $\Theta \times \Psi$.  Hence the
joint probability density function of the observed data satisfies
\begin{equation}
p_{\theta, \psi}(\MAO,R)=\int \int p_{\theta}(\MAO,\MAM,Z)p_\psi(R|\MAO,\MAM,Z)d\MAM dZ.
\label{eq:likelihood}
\end{equation}
Simplifications may occur in \eqref{eq:likelihood} depending on the
sampling design, leading to the three usual types of missingness
(MCAR, MAR and NMAR).  This typology depends on the relations between
the adjacency matrix $Y$, the latent structure $Z$ and the sampling
$R$, so that the missingness is characterized by four directed acyclic
graphs displayed in Figure~\ref{fig:DAGs}.
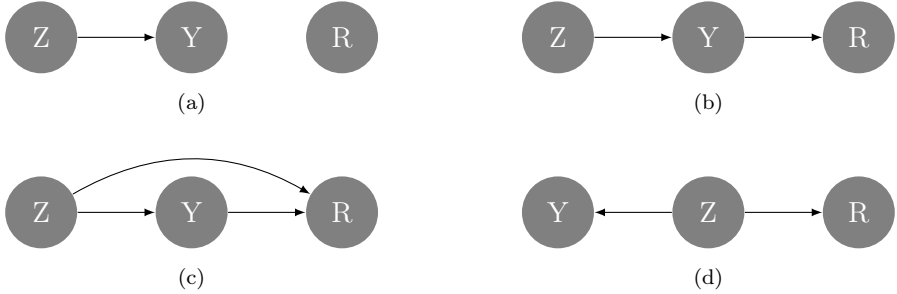
\begin{figure}[htbp!]
  \centering
  \begin{tabular}{@{}c@{\hspace{5em}}c@{}}
    \begin{tikzpicture}
      \tikzstyle{every edge}=[-,>=stealth',shorten >=1pt,auto,thin,draw]
      \tikzstyle{every state}=[draw=none,text=white, font=\normalsize, transform shape]
      \tikzstyle{every node}=[fill=white!50!black]
      \node[state] (Z) at (0,0) {Z};
      \node[state] (\MA) at (2,0) {\MA};
      \node[state] (R) at (4,0) {R};
      \draw[->,>=latex] (Z) -- (\MA);
    \end{tikzpicture}
    &
    \begin{tikzpicture}
      \tikzstyle{every edge}=[-,>=stealth',shorten >=1pt,auto,thin,draw]
      \tikzstyle{every state}=[draw=none,text=white, font=\normalsize, transform shape]
      
      \tikzstyle{every node}=[fill=white!50!black]
      \node[state] (Z) at (0,0) {Z};
      \node[state] (\MA) at (2,0) {\MA};
      \node[state] (R) at (4,0) {R};
      
      \draw[->,>=latex] (Z) -- (\MA);
      \draw[->,>=latex] (\MA) -- (R);
    \end{tikzpicture}
    \\
    \scriptsize (a) & \scriptsize (b) \\[3ex]
    \begin{tikzpicture}
        
        \tikzstyle{every edge}=[-,>=stealth',shorten >=1pt,auto,thin,draw]
      \tikzstyle{every state}=[draw=none,text=white, font=\normalsize, transform shape]
        
        \tikzstyle{every node}=[fill=white!50!black]
        \node[state] (Z) at (0,0) {Z};
        \node[state] (\MA) at (2,0) {\MA};
        \node[state] (R) at (4,0) {R};
        
        \draw[->,>=latex] (Z) -- (\MA);
        \draw[->,>=latex] (\MA) -- (R);
        \draw[->,>=latex] (Z) to[bend left] (R);
        
      \end{tikzpicture}
    &
      \begin{tikzpicture}
        
        \tikzstyle{every edge}=[-,>=stealth',shorten >=1pt,auto,thin,draw]
      \tikzstyle{every state}=[draw=none,text=white, font=\normalsize, transform shape]
        \tikzstyle{every node}=[fill=white!50!black]
        \node[state] (Z) at (0,0) {Z};
        \node[state] (\MA) at (-2,0) {\MA};
        \node[state] (R) at (2,0) {R};
        
        \draw[->,>=latex] (Z) -- (\MA);
        \draw[->,>=latex] (Z) -- (R);
      \end{tikzpicture}
    \\
    \scriptsize (c) & \scriptsize (d) \\
  \end{tabular}    
  \caption{DAGs of relationships between $Y,Z$ and $R$ in the
    framework of missing data for SBM. DAG where $R$ is a parent node
    are not reviewed since the network exists before the sampling
    design acts upon it. The systematic edge between $Z$ and $Y$ is
    due to the definition of the SBM. Note that the DAG $(b)$ may
    correspond to MAR or NMAR samplings.}
  \label{fig:DAGs}
\end{figure}

On the basis of these DAGs, the sampling design is MCAR if
$R \ \indep \ (\MAM, Z, \MAO)$, MAR if
$R \ \indep \ (\MAM,Z) \ | \ \MAO$, and NMAR otherwise. We derive
Proposition~\ref{prop:mar} from these definitions.

\begin{proposition}\label{prop:mar}
  If the sampling is MCAR or MAR then $i)$
  $\argmax_\theta p_{\theta, \psi}(\MAO,R) = \argmax_\theta
  p_{\theta}(\MAO)$ for any $\psi$ such that
  $p_{\theta, \psi}(\MAO,R)\not=0$ and $ii)$ the sampling design
  necessary satisfies DAG $(a)$ or $(b)$.
\label{MAR}
\end{proposition}

\begin{proof} To prove $i)$, if $R$ satisfies MAR conditions, then
  $p_\psi(R|\MAO,\MAM,Z)=p_\psi(R|\MAO)$.  Moreover, $\theta$ and
  $\psi$ lie in a product space so that \eqref{eq:likelihood}
  factorizes into
  $p_{\theta, \psi}(\MAO,R)=p_{\theta}(\MAO)p_\psi(R|\MAO).$ This
  corresponds to the ignorability condition of
  \citet{Rubin1976,Handcock2010}. The proof of $ii)$ is postponed to
  the supplementary materials.
\end{proof}

\subsection{Sampling design examples}
\label{sec:sampling-definition}


\subsubsection{MAR examples}

\begin{definition}[Random-dyad sampling]
  Each dyad $(i,j) \in\dyad$ has the same probability
  $\Pbb(R_{ij}=1)=\rho$ to be observed independently of the others.
\end{definition}
This design  is trivially MCAR because  each dyad is sampled  with the
same probability $\rho$ which does not depend on $Y$.

\begin{definition}[Star and snowball sampling]
  The star sampling consists in selecting uniformly a set of nodes,
  then observing corresponding rows of matrix $\MA$.  Snowball
  sampling is initialized by a star sampling which gives a first
  "wave" of nodes. The second wave is composed by the neighbors of the
  first. Successive waves can then be obtained. The final set of
  observed dyads corresponds to all dyads involving at least one of
  these nodes.
\end{definition}

These two designs are node-centered and MAR. Indeed, selecting nodes
independently in star sampling or in the first wave of snowball
sampling corresponds to MCAR sampling.  Successive waves are then MAR
since they are built on the basis of the previously observed part of
$\MA$.  Expressions of the corresponding distributions
$p_\psi(R|\MAO)$ are given in \citet{Handcock2010}.

\paragraph*{Identifiability of random-dyad and star sampling designs.}

Since random-dyad and star samplings are MCAR, the identifiability is
assessed in two steps by proving the identifiability of, first, the
sampling parameter $\psi = \rho$ and second, the SBM parameters
$\theta=(\alpha,\pi)$ given $\rho$. Our proofs, postponed to the
supplementary materials, follow \citet{celisse2012consistency} who
established the identifiability of the SBM without missing data.
  
\begin{proposition}
\label{thm:ident_mcarParam}
The sampling parameter $\rho > 0$ of random-dyad (resp. star) sampling
is identifiable w.r.t. the sampling distribution.
\end{proposition}
\begin{thm}\label{thm.identifiabilityMAR}
  Let $n\geq 2Q$ and assume that for any $1\leq q \leq Q$, $\rho>0$,
  $\alpha_q >0$ and that the coordinates of $\pi \alpha$ are pairwise
  distinct.  Then, under random-dyad (resp. star) sampling, SBM
  parameters are identifiable w.r.t. the distribution of the observed
  part of the SBM up to label switching.
\end{thm}

\subsubsection{NMAR examples}
\label{sec:nmar_designs}

\begin{definition}[Double standard sampling]
  Let $\rho_{1}, \rho_{0} \in [0,1]$. Double standard sampling
  consists in observing dyads with probabilities
  \begin{equation}
    \Pbb(R_{ij}=1|Y_{ij}=1) = \rho_{1}, \qquad
    \Pbb(R_{ij}=1|Y_{ij}=0) = \rho_{0}. 
  \end{equation}
\end{definition}
Denote
$S^\text{\rm o} = \sum_{(i,j)\in \dyadO} Y_{ij}, \ \bar{S}^\text{\rm
  o} = \sum_{(i,j)\in \dyadO} (1-Y_{ij})$ and similarly for
$S^m, \bar{S}^m$. In this dyad-centered sampling design
satisfying DAG $(b)$, the log-likelihood is
\begin{equation}
\log p_\psi(R|\MA)= S^\text{\rm o} \log \rho_1 + \bar{S}^\text{\rm o}
\log \rho_0 + S^\text{\rm m} \log (1-\rho_1) + \bar{S}^\text{\rm m} \log (1-\rho_0), \quad \text{with } \psi  = (\rho_0,  \rho_1).
\label{eq:logLik_2stand}
\end{equation}

\begin{definition}[Star sampling based on degrees -- Star degree sampling]
\label{def:star_sampling}
Star degree sampling consists in observing all dyads corresponding to
nodes selected with probabilities $\{ \rho_1, \dots, \rho_n \}$ such
that $\rho_{i}={\rm logistic} (a+b D_i)$ for all $i\in\node$ where
$(a,b) \in \mathbb{R}^2$, $D_i = \sum_j Y_{ij}$ and
${\rm logistic}(x) = (1+ e^{-x})^{-1}$.
\end{definition}
In this node-centered sampling design satisfying DAG $(b)$, the
log-likelihood is
\begin{equation}
\log p_\psi(R|\MA) = \sum_{i\in \nodeO} \log \rho_i + \sum_{i\in \nodeM}
\log(1-\rho_i), \qquad \text{with } \psi = (a,b).
\label{eq:logLik_degree}
\end{equation}

\begin{definition}[Class sampling]
  \label{def:class_sampling} Class sampling consists in observing all
  dyads corresponding to nodes selected with probabilities
  $\{ \rho_1, \dots, \rho_Q \}$ such that
  $\rho_q = \Pbb(i\in \nodeO \ | \ Z_{iq}=1)$ for all
  $(i,q) \in \node \times \block $.
\end{definition}
In this node-centered sampling design satisfying DAG $(d)$, the
log-likelihood is
\begin{equation}
  \log p_\psi(R|Z) = \sum_{i\in \nodeO} \sum_{q \in \mathcal{Q}}Z_{iq} \log \rho_{q} + \sum_{i\in \nodeM} \sum_{q \in \mathcal{Q}}Z_{iq}
  \log(1-\rho_{q}), \qquad \text{with } \psi = (\rho_1,\dots,\rho_Q).
  \label{eq:logLik_classes}
\end{equation}

\paragraph*{Identifiability of class sampling.}
Theorem \ref{thm.identifiabilityClass} establishes the identifiability
of the SBM sampled under NMAR class sampling design (see the
supplementary materials for the proof). Note that the identifiability
of the sampling parameters $\psi = (\rho_1,\dots,\rho_Q)$ and of the
SBM parameters must be proved jointly because of the dependence
between the network and the sampling. It is worth mentioning that both
$\alpha_q$ and $\rho_q$ are identifiable and not only their
product. Although somewhat counter-intuitive, this fact is supported
by the inference algorithm for class sampling in
Section~\ref{subsec:nmar_inference}, which weights the recovery of the
latent clusters by taking the unbalanced sampling into account.

\begin{thm}\label{thm.identifiabilityClass}
  Let $n\geq 2Q$ and assume that for any $1\leq q \leq Q$, $\rho_q>0$,
  $\alpha_q>0$, and that the coordinates of $o=\pi \alpha$ and
  $t=(\sum_{k=1}^Q\pi_{1k}\rho_k\alpha_k,\ldots,\sum_{k=1}^Q\pi_{Qk}\rho_k\alpha_k)$
  are pairwise distinct.  Then, under class sampling, SBM and class
  sampling parameters are identifiable w.r.t. the distributions of the
  SBM and the sampling up to label switching.
\end{thm}
 	

\section{Variational Inference}
\label{sec:inference}

Derivations of the practical variational algorithms considerably
change depending on the missing data condition at play.  We start by
MAR to gently introduce the variational principle for SBM, then
develop algorithms in a series of NMAR conditions


\subsection{MAR inference}
\label{sec:MARinference}

By Proposition~\ref{prop:mar} part $(i)$, inference in the MAR case is
conducted on $\MAO$. The EM algorithm is unfeasible since it requires
the evaluation of the conditional mean of the complete log-likelihood
$\Ebb_{Z|\MAO}\left[\log p_{\theta}(\MAO,Z)\right]$ which is
intractable when $\MA$ comes from an SBM. The variational approach
circumvents this limitation by maximizing a lower bound of the
log-likelihood based on an approximation $\tilde p_\tau$ of the true
conditional distribution $p_{\theta}(Z|\MAO)$,
\begin{equation}
  \label{eq:lowerBound}
   \begin{aligned}
    \log p_{\theta}(\MAO) \geq
    J_{\tau,\theta}(\MAO) & \triangleq
    \log (p_{\theta}(\MAO))-\text{KL}[\tilde p_\tau(Z)||p_{\theta}(Z|\MAO)], \\
    & = \mathbb{E}_{\tilde p_\tau} \left[\log(p_{\theta}(\MAO, Z)) \right] 
    -\mathbb{E}_{\tilde p_\tau}[\log \tilde p_\tau(Z)],
  \end{aligned}
\end{equation}
where $\tau$ are some variational parameters and $\text{KL}$ is the
Kullback-Leibler divergence.  The approximated distribution is chosen
so that the integration over the latent variables simplifies by
factorization. Recall from Section~\ref{sec:SBM} that the latent
vectors $\big(Z_{i\sbullet}=(Z_{i1},\ldots, Z_{iQ})\big)_{i\in\node}$ are
independent with a multinomial prior distribution. Thus, in order to
factorize the likelihood in a convenient way, a natural variational
counterpart to $p_\theta (Z|\MAO)$ is
$\tilde p_\tau({Z})= \prod_{i\in\node} m(Z_{i\sbullet};\tau_i)$, where
$\tau_i=(\tau_{i1},\ldots, \tau_{iQ})$, and $m(\cdot;\tau_{i})$ is the
multinomial probability density function with parameters $\tau_i$. The
VEM sketched in Algorithm \ref{algo:vem:mar} consists in alternatively
maximizing $J$ w.r.t.  $\tau = \{\tau_1,\dots,\tau_n\}$ (the
variational E-step) and w.r.t.  $\theta$ (the M-step). The two
maximization problems are solved straightforwardly following
\citet{daudin2008mixture}:
\begin{enumerate}
\item   The  parameters $\theta=(\alpha,  \pi)$ maximizing  $J_\theta(\MAO)$
  when $\tau$ is held fixed are
  \begin{equation}\nonumber
    \hat{\alpha}_q=\frac{\sum_{i\in \nodeO}\hat{\tau}_{iq}}{\card{\nodeO}}, \qquad
    \hat{\pi}_{q\ell}=\frac{\sum_{(i,j)\in
        \dyadO}\hat{\tau}_{iq}\hat{\tau}_{j\ell}\MA_{ij}}{\sum_{(i,j)\in
        \dyadO
      }\hat{\tau}_{iq}\hat{\tau}_{j\ell}}.
  \end{equation}
\item The variational parameters $\tau$ maximizing $J_\tau(\MAO)$ when
  $\theta$ is held fixed are obtained with the following fixed point
  relation:
  \begin{equation}\nonumber
    \hat{\tau}_{iq}\propto \alpha_{q} \left( \prod_{(i,j)\in\dyadO}
      \prod_{\ell\in\block} b(\MA_{ij}; \pi_{q\ell})^{\hat{\tau}_{j\ell}}\right),
  \end{equation}
  where $b(x,\pi)=\pi^{x}(1-\pi)^{1-x}$ the Bernoulli probability
  density function.
\end{enumerate}

\begin{algorithm}[H]
  \SetSideCommentLeft
  \DontPrintSemicolon
  \KwSty{Initialization:} Set up $\tau^{(0)}$ with some clustering algorithm\;

  \Repeat{$\left\|\theta^{(h+1)} - \theta^{(h)}\right\| < \varepsilon$}{
    \begin{equation}\nonumber
      \begin{array}{lcl@{\hspace{1cm}}r}
       \theta^{(h+1)} & = & \displaystyle\argmax_{\theta}J\left(\MAO; \tau^{(h)},\theta
      \right) & \textbf{M-step}\\
      \tau^{(h+1)} & = & \displaystyle\argmax_\tau J\left(\MAO; \tau,\theta^{(h+1)} \right) & \textbf{variational E-step} \\   
      \end{array}
    \end{equation}
  }
 \caption{Variational EM for MAR inference in SBM}
 \label{algo:vem:mar}
\end{algorithm}

Algorithm~\ref{algo:vem:mar} generates a sequence
$\{ \tau^{(h)},\theta^{(h)}; h\geqslant 0\}$ with increasing
$J(\MAO; \tau^{(h)}, \theta^{(h)})$. Since there is no guarantee for
convergence to the global maximum, we run the algorithm from several
different initializations to finally retain the best solution.
 
\paragraph*{Model selection of the number of blocks.} The Integrated
Classification Likelihood (ICL) criterion of \citet{Biernacki2000} is
relevant for latent variable models where the likelihood -- and thus
BIC -- is intractable. \citet{daudin2008mixture} derive a variational
ICL for the SBM which we adapt to missing data conditions: if
$\hat{\theta}=\argmax \log p_{\theta}(\MAO,Z)$ then
\begin{equation}\nonumber
  \mathrm{ICL}(Q)   =    -2   \mathbb{E}_{\tilde   p_\tau}\left[\log
    p_{\hat\theta}(\MAO,  Z   ;  Q)\right]   +  \frac{Q(Q+1)}{2}\log
  \card{\dyadO} + (Q-1)\log \card{\nodeO}. 
\end{equation}
Note that each dyad is only counted once since we work with symmetric
networks.

\subsection{NMAR inference: the general case}
\label{sec:NMARinference}

In contrast to the MAR case, conducting inference on the observed
dyads only may bias the estimates in the NMAR case.  In fact, all
observed data (including the sampling matrix $R$ in addition to
$\MAO$) must be taken into account. The likelihood of the observed
data is thus $\log p_{\theta,\psi}(\MAO,R)$ and the corresponding
completed likelihood has the following decomposition:
\begin{equation}
\log  p_{\theta, \psi}(\MAO,R,\MAM,Z)  = \log  p_\psi(R|\MAO,\MAM,Z) +
\log p_\theta(\MAO,\MAM,Z),
 \label{eq:complete_loglik_nmar}
\end{equation}
where an explicit form of $p_\psi(R|\MAO,\MAM,Z)$ requires further
specification of the sampling.  The joint distribution
$p_\theta(\MAO,\MAM,Z)$ has a form similar to the MAR case. Now, the
approximation is required both on latent blocks $Z$ and missing dyads
$\MAM$ to approximate $p_{\theta}(Z,\MAM|\MAO)$.  We suggest a
variational distribution where complete independence is forced on $Z$
and $\MAM$, using a multinomial (resp. Bernoulli ) distribution for
$Z$ (resp. for $\MAM$):
\begin{equation}
  \label{eq:approx_nmar}
  \tilde p_{\tau,\nu} (Z,\MAM)= \tilde p_{\tau} (Z) \ \tilde p_{\nu} (\MAM)
  = \prod_{i\in\node} m(Z_{i\cdot};\tau_{i}) \prod_{(i,j) \in
    \dyadM} b(Y_{ij};\nu_{ij}),
\end{equation}
where $\tau$ and $\nu = \{\nu_{ij}, (i,j)\in\dyadM\}$ are two sets of
variational parameters respectively associated with $Z$ and $\MAM$.
This leads to the following lower bound for
$\log p_{\theta,\psi}(\MAO,R)$:
\begin{equation}
  \nonumber
    J_{\tau,\nu,\theta,\psi}(\MAO,R)
    =  \Ebb_{\tilde       p_{\tau,\nu}}       \left[\log
      p_{\theta,\psi}(\MAO,R,\MAM,Z)\right]
    -\Ebb_{\tilde p_{\tau,\nu}}\left[ \log \tilde p_{\tau,\nu}(Z,\MAM)\right].
\end{equation}
By  means  of  Decomposition  \eqref{eq:complete_loglik_nmar}  of  the
completed                  log-likelihood,                 variational
approximation~\eqref{eq:approx_nmar} and entropies  of multinomial and
Bernoulli distributions, one has
\begingroup
\setlength\abovedisplayskip{0pt}
\begin{multline}
  J_{\tau,\nu,\theta,\psi}(\MAO,R)
  = \Ebb_{\tilde p_{\tau,\nu}} \left[\log p_\psi(R | \MAO,\MAM,Z)\right] \\ 
  +\sum_{(i,j)\in\dyadO}\sum_{(q,\ell)\in\block^2}\tau_{iq}\tau_{j\ell}\log
  b(\MA_{ij},\pi_{q\ell})
  +\sum_{(i,j)\in\dyadM}\sum_{(q,\ell)\in\block^2}\tau_{iq}\tau_{j\ell}\log
  b(\nu_{ij},\pi_{q\ell}) \\
  + \sum_{i\in\node}\sum_{q\in\block}\tau_{iq} \log (\alpha_{q}/\tau_{iq})
  -\sum_{(i,j)\in \dyadM} \nu_{ij}\log (\nu_{ij}) + (1-\nu_{ij}) \log (1-\nu_{ij}).
\label{eq:var:approx:nmar}
\end{multline}
\endgroup
In \eqref{eq:var:approx:nmar},
$\Ebb_{\tilde p_{\tau,\nu}} \left[\log p_\psi(R | \MAO,\MAM,Z)\right]$
can be integrated over the variational distribution
$\tilde p_{\tau,\nu}(Z,\MAM)$, as expected. The practical computations
depend on the sampling design.

The general VEM algorithm used to maximize \eqref{eq:var:approx:nmar}
is sketched in its main lines in Algorithm~\ref{algo:vem:nmar}.  Both
the E-step and the M-step split into two parts: the maximization must
be performed on the SBM parameters $\theta$ and the sampling design
parameters $\psi$ respectively.  The variational E-step is performed
on the parameters $\tau$ of the latent block $Z$ and on the parameters
$\nu$ of the missing data $\MAM$ .

\begin{algorithm}[H]
  \SetSideCommentLeft
  \DontPrintSemicolon
  \KwSty{Initialisation:} set up $\tau^{(0)}$, $\bnu^{(0)}$ and $\psi^{(0)}$\;
  \Repeat{$\left\|\theta^{(h+1)} - \theta^{(h)}\right\| < \varepsilon$}{
    \begin{equation}\nonumber
      \begin{array}{lclr}
        \theta^{(h+1)}&=&\argmax_{\theta}       J\left(\MAO,R;       \
                          \tau^{(h)},\nu^{(h)},\psi^{(h)},\theta
                          \right) & \text{\textbf{M-step a)}}\\
        \psi^{(h+1)}&=&\argmax_{\psi}          J\left(\MAO,R;                               \
                        \tau^{(h)},\nu^{(h)},\psi,\theta^{(h+1)}
                        \right) & \text{\textbf{M-step b)}}\\
        \tau^{(h+1)}&=&\argmax_{\tau} J\left(\MAO,R; \ \tau,\nu^{(h)},\psi^{(h+1)}, \theta^{(h+1)} \right) & \text{\textbf{VE-step  a)}}\\
        \nu^{(h+1)}&=&\argmax_{\nu} J\left(\MAO,R; \tau^{(h+1)},\nu,\psi^{(h+1)},\theta^{(h+1)} \right) & \text{\textbf{VE-step b)}}
      \end{array}
    \end{equation}
  }
 \caption{Variational EM for NMAR inference in SBM}
 \label{algo:vem:nmar}
\end{algorithm}

Interestingly, resolution of the two steps concerned with the
optimization of the parameters related with the SBM -- that is to say,
$\theta$ and $\tau$ -- can be stated almost independently of any
further specification of the sampling design.
\begin{proposition}\label{prop:nmar_common}
  Consider   the  lower   bound  $J_{\tau, \nu, \theta, \psi}(\MAO,R)$   given  by
  \eqref{eq:var:approx:nmar}.
  \begin{enumerate}
  \item The parameters $\theta=(\alpha, \pi)$ maximizing 
    \eqref{eq:var:approx:nmar} when all others are held fixed are
    \begin{equation}\nonumber
      \hat{\alpha}_q=\frac{1}{n}\sum_{i\in\node} \hat{\tau}_{iq}, \qquad
      \hat{\pi}_{q\ell}=\frac{\sum_{(i,j)\in
          \dyadO}\hat{\tau}_{iq}\hat{\tau}_{j\ell}\MA_{ij}           +
        \sum_{(i,j)\in
          \dyadM}\hat{\tau}_{iq}\hat{\tau}_{j\ell}\hat{\nu}_{ij}}{\sum_{(i,j)\in\dyad}\hat{\tau}_{iq}\hat{\tau}_{j\ell}}.      
    \end{equation}
  \item  The optimal  $\tau$  in  \eqref{eq:var:approx:nmar} when  all
    other parameters are held fixed verifies
    \begin{equation}\nonumber
      \hat{\tau}_{iq}\propto \lambda_{iq} \alpha_{q} \left( \prod_{(i,j)\in \dyadO}
        \prod_{\ell\in\block}                                       b(\MA_{ij};
        \pi_{q\ell})^{\hat{\tau}_{j\ell}}\right)
      \left( \prod_{(i,j)\in \dyadM}
        \prod_{\ell\in\block} b(\nu_{ij}; \pi_{q\ell})^{\hat{\tau}_{j\ell}}\right)
    \end{equation}
    with $\lambda_{iq}$ a simple constant depending on the sampling design.
  \end{enumerate}
\end{proposition}
\begin{proof} These results are simply obtained by differentiation of
  \eqref{eq:var:approx:nmar}.
\end{proof}

The two steps concerned with $\psi$ and $\nu$ are specific to the
sampling designs used to describe $R$. Further details are provided
below for the designs presented in Section~\ref{sec:nmar_designs}.

\subsection{NMAR: specificities related to the choice of the sampling}
\label{subsec:nmar_inference}

In light of Figure~\ref{fig:DAGs}, NMAR conditions specified
by DAGs $(b), (c)$ or $(d)$ induce different simplifications for the
conditional distribution of the sampling design $R$:
\begin{description}
 \item[DAG (b)]  $p_{\psi}(R| \MAO,\MAM, Z)=p_{\psi}(R| \MAO,\MAM)$,
 \item[DAG (c)]  $p_{\psi}(R| \MAO,\MAM, Z)=p_{\psi}(R| \MAO,\MAM,Z)$,
 \item[DAG (d)] $p_{\psi}(R| \MAO,\MAM, Z)=p_{\psi}(R| Z)$.
\end{description}
This induces different evaluations of
$\Ebb_{\tilde p_{\tau,\nu}} \left[\log
  p_{\theta,\psi}(\MAO,R,\MAM,Z)\right]$ in the lower
bound~\eqref{eq:var:approx:nmar} for double standard sampling, star
degree sampling and class sampling. We obtain below explicit formulas
of $\psi$ and $\nu$ by differentiation of the corresponding
variational lower bounds. The computations are tedious but
straightforward and thus eluded in the following.


\paragraph*{\bf Double-standard sampling.}  Let
$s^\text{\rm m} = \sum_{(i,j)\in \dyadM} \nu_{ij}, \ \bar{s}^\text{\rm
  m} = \sum_{(i,j)\in \dyadM} (1-\nu_{ij})$ be the variational
counterparts of $S^\text{\rm m}$ and $\bar{S}^\text{\rm m}$.  From
\eqref{eq:logLik_2stand} we have
\begin{equation}\nonumber
  \mathbb{E}_{\tilde p} \log p_{\psi}(R|\MA) =
  S^\text{\rm o} \log \rho_1 + \bar{S}^\text{\rm o}
  \log \rho_0 + s^\text{\rm m}  \log (1-\rho_1) + \bar{s}^\text{\rm m}
  \log (1-\rho_0).
\end{equation}
\begin{proposition}[double standard sampling]\label{prop:nmar_doublestand}~\\[-2ex]
  \begin{enumerate}
  \item   The   parameters   $\psi   =   (\rho_0,\rho_1)$   maximizing
    \eqref{eq:var:approx:nmar}  when  all  others  are  held
    fixed are
    \begin{equation}
      \label{eq:rho_01}
      \hat{\rho}_0  = \frac{\bar{S}^\text{\rm  o}}{\bar{S}^\text{\rm o}  +
        \bar{s}^\text{\rm m}}, \qquad
      \hat{\rho}_1 = \frac{S^\text{\rm o}}{S^\text{\rm o} + s^\text{\rm m}}.
    \end{equation}
    
  \item The optimal $\nu$ in \eqref{eq:var:approx:nmar} when all
    other parameters are held fixed are
    \begin{equation}\nonumber
      \hat{\nu}_{ij} = \mathrm{logistic} \left( \log \left( \frac{1-\rho_1}{1-\rho_0} \right) + \sum_{(q,\ell)\in\block^2} \tau_{iq}\tau_{j\ell} \log\left(\frac{\pi_{q\ell}}{1-\pi_{q\ell}}\right) \right).
    \end{equation}
  \end{enumerate}
  Moreover, $\lambda_{iq}  = 1 \ \forall  (i,q) \in \node\times\block$
  for optimization of $\tau$ in Proposition~\ref{prop:nmar_common}.b).
\end{proposition}

\paragraph*{\bf Class sampling.}  According to
\eqref{eq:logLik_classes} we have
\begin{equation}\nonumber
  \mathbb{E}_{\tilde{p}} \log p_\psi(R|\MA)  = \sum_{i\in \nodeO}
  \sum_{q\in\block} \tau_{iq} \log(\rho_q) + 
  \sum_{i\in \nodeM}\sum_{q\in\block} \tau_{iq}\log(1-\rho_q).
\end{equation}

\begin{proposition}[class sampling]\label{prop:nmar_class}~\\[-2ex]
  \begin{enumerate}
  \item  The  parameters  $\psi  =  (\rho_1...\rho_Q)$  maximizing
    \eqref{eq:var:approx:nmar}  when  all  others  are  held
    fixed are
    \begin{equation}
      \label{eq:rho_q}
      \hat{\rho}_q    =   \frac{\sum_{i\in\nodeO}\tau_{iq}}{\sum_{i\in\node}\tau_{iq}}.
    \end{equation}
  \item The optimal $\nu$ in \eqref{eq:var:approx:nmar} when all
    other parameters are held fixed verify
    \begin{equation}\nonumber
      \hat{\nu}_{ij} = \mathrm{logistic} \left( \sum_{(q,\ell)\in\block^2} \tau_{iq}\tau_{j\ell} \log\left(\frac{\pi_{q\ell}}{1-\pi_{q\ell}}\right) \right).
    \end{equation}
    
  \end{enumerate}
  Moreover   $\lambda_{iq}  =   \rho_q^{\1_{\{i\in\nodeO\}}}(1-\rho_q)^{^{\1_{\{i\in\nodeM\}}}}$  for
  optimization of $\tau$ in Proposition~\ref{prop:nmar_common}.b).
\end{proposition}

\paragraph*{\bf Star degree sampling.}  From
Expression~\eqref{eq:logLik_degree} of the  likelihood, one has
\begin{equation}\nonumber
  \mathbb{E}_{\tilde{p}} \log p_{\psi}(R|\MA)
  = -\sum_{i\in\nodeM} \left(a+b\tilde{D}_i \right)  + \sum_{i\in\node}  \Ebb_{\tilde{p}}\left[ -\log(1+e^{-(a+bD_i)}) \right],
\end{equation}
where
$\tilde{D}_i = \mathbb{E}_{\tilde{p}}\left[D_i \right] = \sum_{i
  \in\nodeM} \nu_{ij} + \sum_{i \in\nodeO} \MA_{ij}$ is the
approximation of the degrees.  Because
$\mathbb{E}_{\tilde{p}}\left[ -\log(1+e^{-(a+bD_i)}) \right]$ has no
explicit form, an additional variational approximation is needed
\citep{jordan1998introduction}. This technique was recently used in
random graph framework \citep{LatoucheRobinOuadah2017}.  It relies on
the following approximation of the logistic function:
\begin{equation}
\label{eq:taylor_expansion}
g(x) \geq g(\zeta) + \frac{x-\zeta}{2} + h(\zeta)(x^2
- \zeta^2), \quad h(\zeta) =  \frac{-1}{2\zeta}\left[ \text{logistic}(\zeta)  - \frac{1}{2}
\right]
\end{equation}
for all $(x,\zeta) \in \mathbb{R}\times\mathbb{R}^{+}$.  This leads to a
lower bound of the initial lower bound:
\begin{equation}
  \label{eq:lower_bound_degree}
  \log p_{\theta,\psi}(\MAO,R)  \geq
  J_{\tau,\nu,\theta,\psi}(\MAO,R) \geq
  J_{\tau,\nu,\zeta,\theta,\psi}(\MAO,R),
\end{equation}
with $\zeta = \left( \zeta_i, i\in\node\right)$ such that $\zeta_i>0$
is an additional set of variational parameters used to approximate
$-\log(1+e^{-x})$. The second lower bound
$J_{\tau,\nu,\zeta,\theta,\psi}$ is derived from
Equation~\eqref{eq:taylor_expansion} and given in the supplementary
materials for completeness .  At the end of the day, we have an
additional set of variational parameters to optimize, and a
corresponding additional step in
Algorithm~\ref{algo:vem:nmar}. Expression of all the parameters
specific to star degree sampling by differentiating
$J_{\tau,\nu,\zeta,\theta,\psi}$.

\begin{proposition}[star degree sampling]\label{prop:nmar_degree} Let
  $\widetilde{D^2_i} =\Ebb_{\tilde{p}}\left[ {D_i}^2 \right]$ and
  $\tilde{D}_{k}^{-\ell} = \tilde{D}_k - \nu_{k\ell}$.
  \begin{enumerate}
  \item The parameters $\psi = (a,b)$ maximizing
    $J_{\tau,\nu,\zeta,\theta,\psi}(\MAO,R)$ when others are held
    fixed are 
\begingroup
\setlength\abovedisplayskip{0pt}
    \begin{align*}
      \hat{b} &= \frac{  2\left(\frac{n}{2}-\card{\nodeM} \right)\sum_{i\in\node} (h(\zeta_i) \tilde{D}_i)  -
\left(\frac{1}{2}\sum_{i\in\node} \tilde{D}_i-\sum_{i\in\nodeM} \tilde{D}_i \right)
\times    \sum_{i\in\node}   h(\zeta_i)}{    2\sum_{i\in\node}(
  h(\zeta_i) \widetilde{D_i^2}) \times  \sum_{i\in\node} h(\zeta_i)
                - {\left( 2\sum_{i\in\node} h(\zeta_i) \tilde{D}_i \right)}^2  }, \\
      \hat{a} &= -\frac{ \hat{b} \sum_{i\in\node} \left(h(\zeta_i)\tilde{D}_i \right) + \frac{n}{2}- \card{\nodeM}}{ \sum_{i\in\node} h(\zeta_i) }.
    \end{align*}
\endgroup
  \item The parameters $\zeta$ maximizing $J_{\tau,\nu,\zeta,\theta,\psi}(\MAO,R)$  when  others are  held
    fixed are
    \begin{equation}\nonumber
      \hat{\zeta}_{i} = \sqrt{a^2 + b^2\widetilde{D_i^2} + 2ab\tilde{D}_i}, \ \forall i \in \mathcal{N}.
    \end{equation}
  \item The optimal  $\nu$ in $J_{\tau,\nu,\zeta,\theta,\psi}(\MAO,R)$
    when all other parameters are held fixed verify
\begingroup
\setlength\abovedisplayskip{0pt}
    \begin{multline}
      \hat{\nu}_{ij}
      = \mathrm{logistic} \Bigg( \sum_{(q,\ell)\in\block^2} \tau_{iq}\tau_{j\ell} \log\left(\frac{\pi_{q\ell}}{1-\pi_{q\ell}}\right)
      - b \\
      + 2h(\zeta_i)\left( ab + b^2(1+\tilde{D}_{i}^{-j}) \right) + 2h(\zeta_j)\left( ab + b^2(1+\tilde{D}_{j}^{-i}) \right) \Bigg).
    \end{multline}
\endgroup
  \end{enumerate}
  Moreover, $\lambda_{iq} = 1 \ \forall (i,q) \in \node\times\block$
  for optimization of $\tau$ in Proposition~\ref{prop:nmar_common}.b).
\end{proposition}

\paragraph*{Model selection.} In NMAR cases, ICL can be useful not
only to select the appropriate number of blocks but also for selecting
the most appropriate sampling design when it is unknown.  Contrary to
the MAR case, ICL is no longer a straightforward generalization of
\citet{daudin2008mixture}.  Indeed, the complete likelihood and thus
the penalization needs to take into account the sampling design. Let
us consider a model with $Q$ blocks and a sampling design with $K$
parameters (\textit{i.e.}  the dimension of $\psi$).  The ICL
criterion is a Laplace approximation of the complete likelihood
$p(\MAO, \MAM, R, Z| Q, K)$ with $p(\theta,\psi|Q,K)$ the prior
distributions on the parameters such that
\begin{equation}\nonumber
  p(\MAO,  \MAM, R,  Z| Q,K)  = \int_{\Theta  \times \Psi}
  p_{\theta, \psi}(\MAO, \MAM, R, Z|Q, K) p(\theta, \psi|Q,K)d\theta d\psi.
\end{equation}

\begin{proposition}
\label{prop:ICL_NMAR}
For a model with $Q$ blocks, a sampling design with a vector of
parameters $\psi\in\mathbb{R}^K$ and
$(\hat{\theta},\hat{\psi})=\argmax_{(\theta, \psi)}\log p_{\theta,
  \psi}(\MAO,\MAM, R, Z)$, then
\begingroup
\setlength\abovedisplayskip{0pt}
\begin{gather*}
  \mathrm{ICL}(Q)        =       -2\mathbb{E}_{\tilde        p_{\tau,\nu};
    \hat{\theta},\hat{\psi}}\left[\log p_{\hat{\theta},\hat{\psi}}(\MAO,\MAM, R, Z | Q, K)\right] + \mathrm{pen}_{\text{ICL}}(Q),\\[2ex]
  \mathrm{pen}_{\text{ICL}} = \left\{
    \begin{array}{ll}
      \left(K + \frac{Q(Q+1)}{2}\right)\log \left(\frac{n(n-1)}{2}\right) + (Q-1)\log (n) &  \text{for dyad-centered sampling} \\
      \frac{Q(Q+1)}{2}\log\left( \frac{n(n-1)}{2} \right)+ (K + Q-1)\log (n) & \text{for node-centered sampling}\\
    \end{array}
  \right.
\end{gather*}
\endgroup
\end{proposition}
Note that an ICL criterion for MAR sampling designs can be constructed in the same fashion for the purpose of comparison with NMAR sampling designs.


\section{Simulation study}
\label{sec:simulations}

In this section, we illustrate the relevance of our approaches on
network data simulated under the SBM and sampled under MAR and NMAR
conditions.  The quality of the inference is assessed by computing the
distance between the estimated and the true connectivity matrices
$\pi$ in terms of Frobenius norm.  The quality of the clustering
recovery is measured with the adjusted Rand index
\citep[ARI,][]{rand1971} between the true classification and the
clustering obtained by maximum posterior probabilities for each
$\tau_i$.

\subsection{MAR condition}
\label{sec:MARsimulations}

Algorithm~\ref{algo:vem:mar} for MAR samplings is tested on
affiliation networks with $3$ blocks.  The number of blocks is assumed
to be known. For this topology the probability of connection within a 
block is $\eta$ and is ten times stronger than the
probability of connections between nodes from different blocks. We
generate networks with $n = 200$ nodes and marginal probabilities of
belonging to blocks $\alpha = (1/3, 1/3, 1/3)$.  The sampling design
is chosen as a random-dyad sampling with a varying $\rho$. The
difficulty is controlled by two parameters: the sampling effort $\rho$
and the overall connectivity in matrix $\pi$, defined by
$c=\sum_{q\ell} \alpha_{q}\alpha_{\ell} \pi_{q\ell}$, which is
directly related to the choice of $\eta$: the lower the $\eta$, the
sparser the network and the harder the inference. The simulation is
repeated 500 times for each configuration $(c,\rho)$. Figure
\ref{fig:MAR_plot} displays the results in terms of estimation of
$\pi$ and of classification recovery, for varying connectivity $c$ and
sampling effort $\rho$.  Our method achieves good performances even
with a low sampling effort provided that the connectivity is not too
low.
\begin{figure}[htbp!]
  \centering
  \begin{tabular}{@{}c@{}c@{}r}
    $\|\hat{\pi} - \pi \|_F / \|\pi\|_F$
    & Adjusted Rand Index & \\
    \includegraphics[width=.45\textwidth]{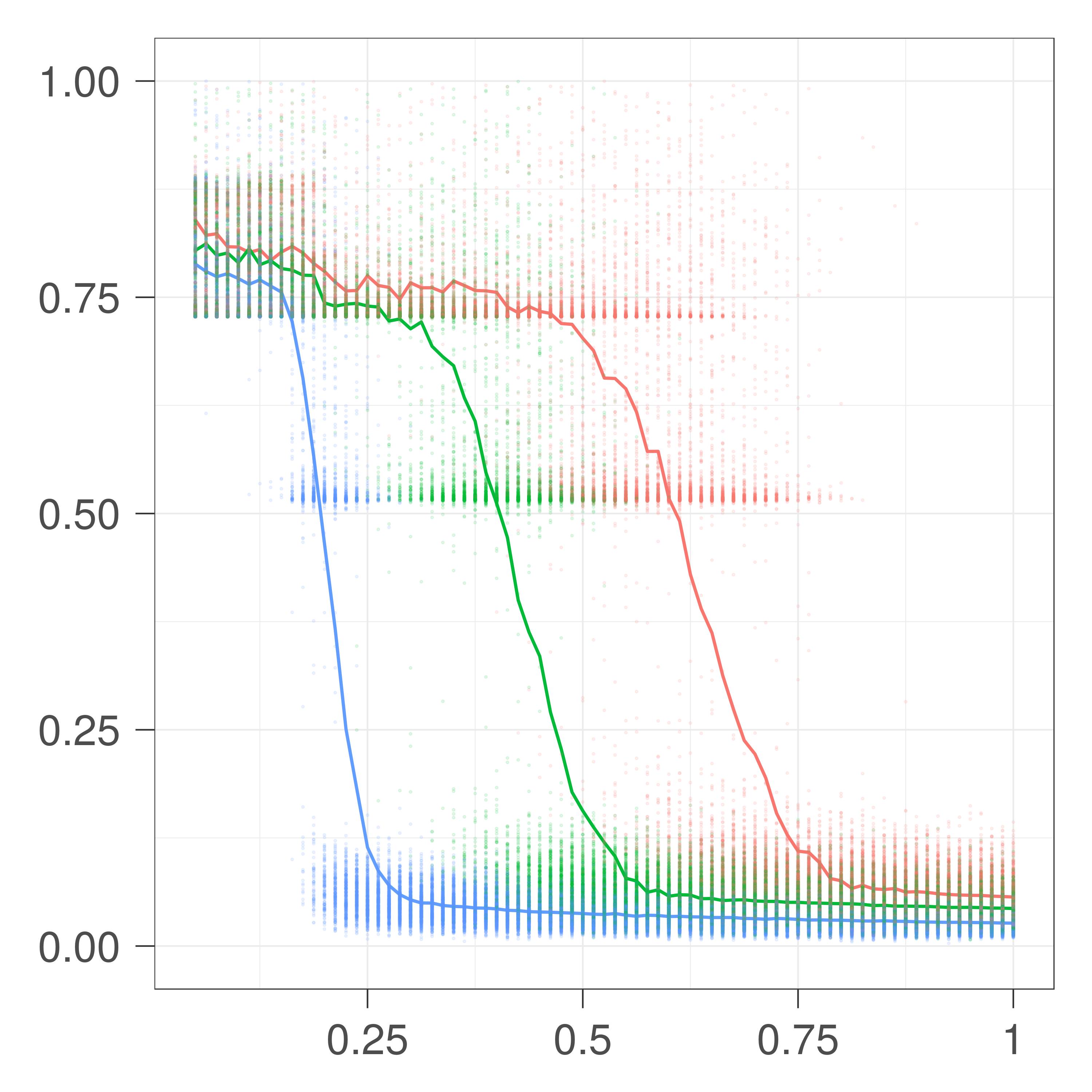} 
    & \includegraphics[width=.45\textwidth]{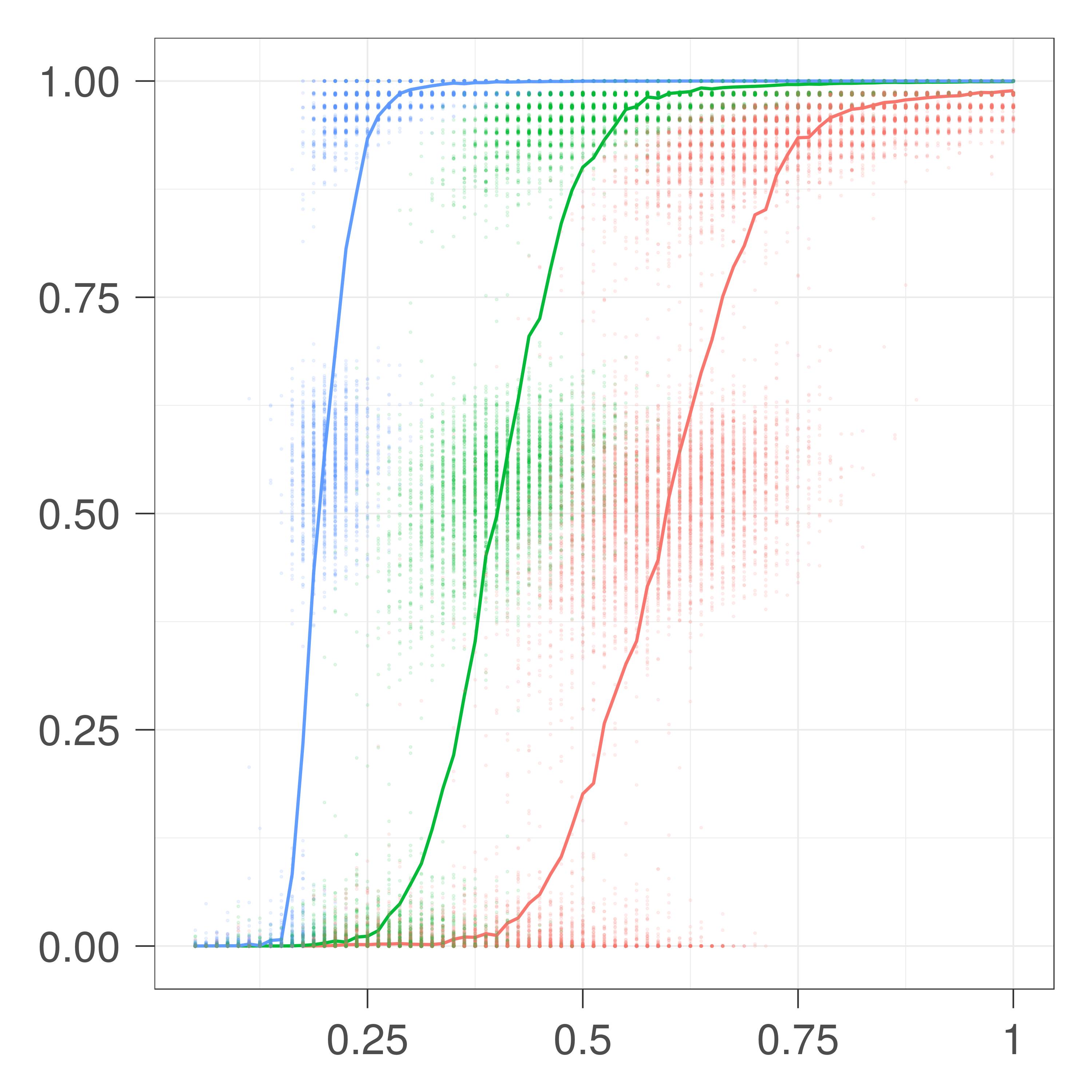} 
    & \raisebox{3.5\height}{\includegraphics[width=.12\textwidth]{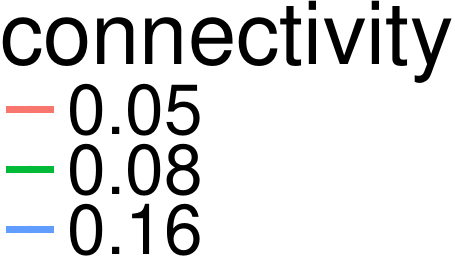}} \\
    \multicolumn{2}{c}{proportion of observed dyads} & \\
  \end{tabular}  \caption{Estimation error  of  $\pi$  and Adjusted  Rand
    Index averaged over 500
    simulations in the  MAR setting. The adjacency matrix $\MA$ is
    generated  under   random-dyad  sampling  strategy   for  various
    connectivity $c = \sum_{q\ell} \alpha_{q}\alpha_{\ell} \pi_{q\ell}$. }
\label{fig:MAR_plot}
\end{figure}

\subsection{NMAR condition}
\label{sec:NMARsimulations}

Under NMAR conditions we conduct an extensive simulation study by
considering various network topologies (namely affiliation, star and
bipartite), the connectivity matrix of which are given in
Figure~\ref{fig:Topologies}.  We use a common tuning parameter
$\epsilon$ to control the connectivity of the networks in each
topology: the lower the $\epsilon$, the more contrasted the topology.
\begin{figure}[htbp!]
  \centering
  \begin{small}
  \begin{subfigure}[b]{.25\textwidth}
    \centering
    $\begin{pmatrix}
      1-\epsilon&\epsilon&\epsilon\\ 
      \epsilon&1-\epsilon&\epsilon\\ 
      \epsilon&\epsilon&1-\epsilon\\ 
    \end{pmatrix}$
    \caption{affiliation}
  \end{subfigure}
  \hspace{2ex}
  \begin{subfigure}[b]{.32\textwidth}
    \centering
    $\begin{pmatrix}
      1-\epsilon&1-\epsilon&0&0\\ 
      1-\epsilon&0&\epsilon&0\\
      0&\epsilon&1-\epsilon&1-\epsilon\\
      0&0&1-\epsilon&0\\ 
    \end{pmatrix}$
    \caption{star}
  \end{subfigure}
  \hspace{2ex}
  \begin{subfigure}[b]{.32\textwidth}
    \centering
    $\begin{pmatrix}
      \epsilon&1-\epsilon&\epsilon&\epsilon\\ 
      1-\epsilon&\epsilon&\epsilon&\epsilon\\
      \epsilon&\epsilon&\epsilon&1-\epsilon\\
      \epsilon&\epsilon&1-\epsilon&\epsilon\\ 
    \end{pmatrix}$
    \caption{bipartite}
  \end{subfigure}
  \end{small}
  \caption{Matrix $\pi$ in different topologies with inter/intra
    block probabilities.}
  \label{fig:Topologies}
\end{figure}

Among the three schemes developed in Section~\ref{sec:nmar_designs},
we investigate thoroughly the double standard sampling, for which we
exhibit a large panel of situations where the gap is large between the
performances of the algorithms designed for MAR or NMAR cases. Other
sampling designs are explored in the supplementary materials.

Simulated networks have $n=100$ nodes, with $\epsilon$ varying in
$\{0.05, 0.15, 0.25\}$.  Prior probabilities $\alpha$ are chosen
specifically for affiliation, star and bipartite topologies,
respectively $(1/3,1/3,1/3)$, $(1/6, 1/3, 1/6, 1/3)$ and
$(1/4,1/4,1/4,1/4)$.  The exploration of the sampling parameters
$\psi=(\rho_0,\rho_1)$ is done on a grid $[0.1,0.9] \times [0.1,0.9]$
discretized by steps of $0.1$.  Algorithm~\ref{algo:vem:nmar} is
initialized with several random initializations and spectral
clustering.

In Figure \ref{fig:simu_twoStd}, the estimation error is represented
as a function of the difference between the sampling design parameters
$(\rho_0,\rho_1)$: the closer this difference to zero, the closer to
the MAR case.  As expected, Algorithm \ref{algo:vem:mar} designed for
MAR only performs well when $\rho_1-\rho_0\approx 0$. Algorithm
\ref{algo:vem:nmar} designed for NMAR double-standard sampling shows
relatively flat curves which means that its performances are roughly
constant no matter the sampling condition.
\begin{figure}[htbp!]
  \centering
  \begin{tabular}{@{}c@{}c@{\hspace{.03mm}}c@{\hspace{1ex}}c@{}}    
    & {\small $\|\hat{ \pi} - \pi \|_F / \| \pi\|_F$} & \small \textbf{ARI}$(Z, \hat{Z})$ & \\
    \rotatebox{90}{\hspace{1.5cm} \small affiliation} &
     \includegraphics[width=.45\textwidth]{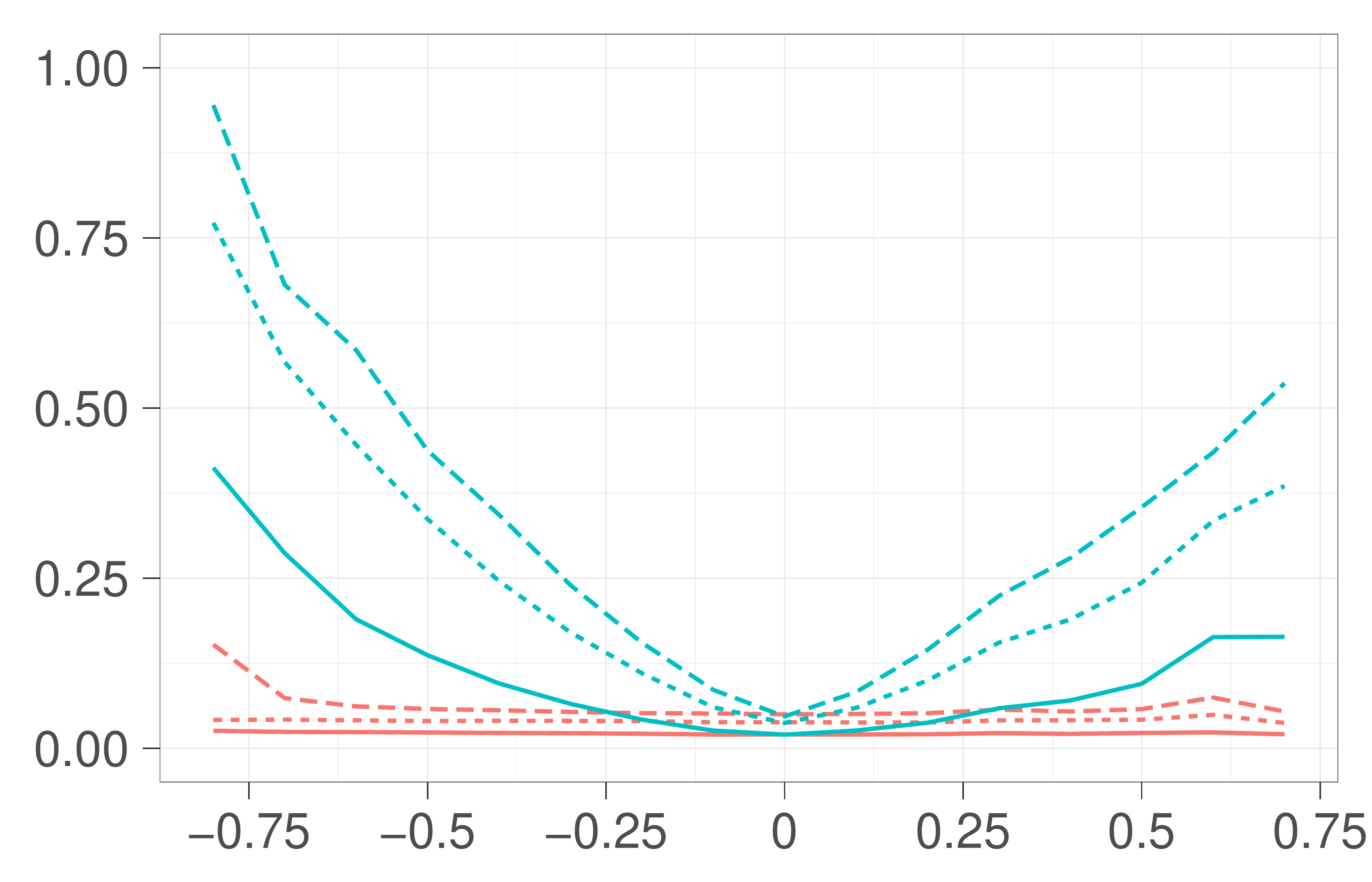}& 
     \includegraphics[width=.45\textwidth]{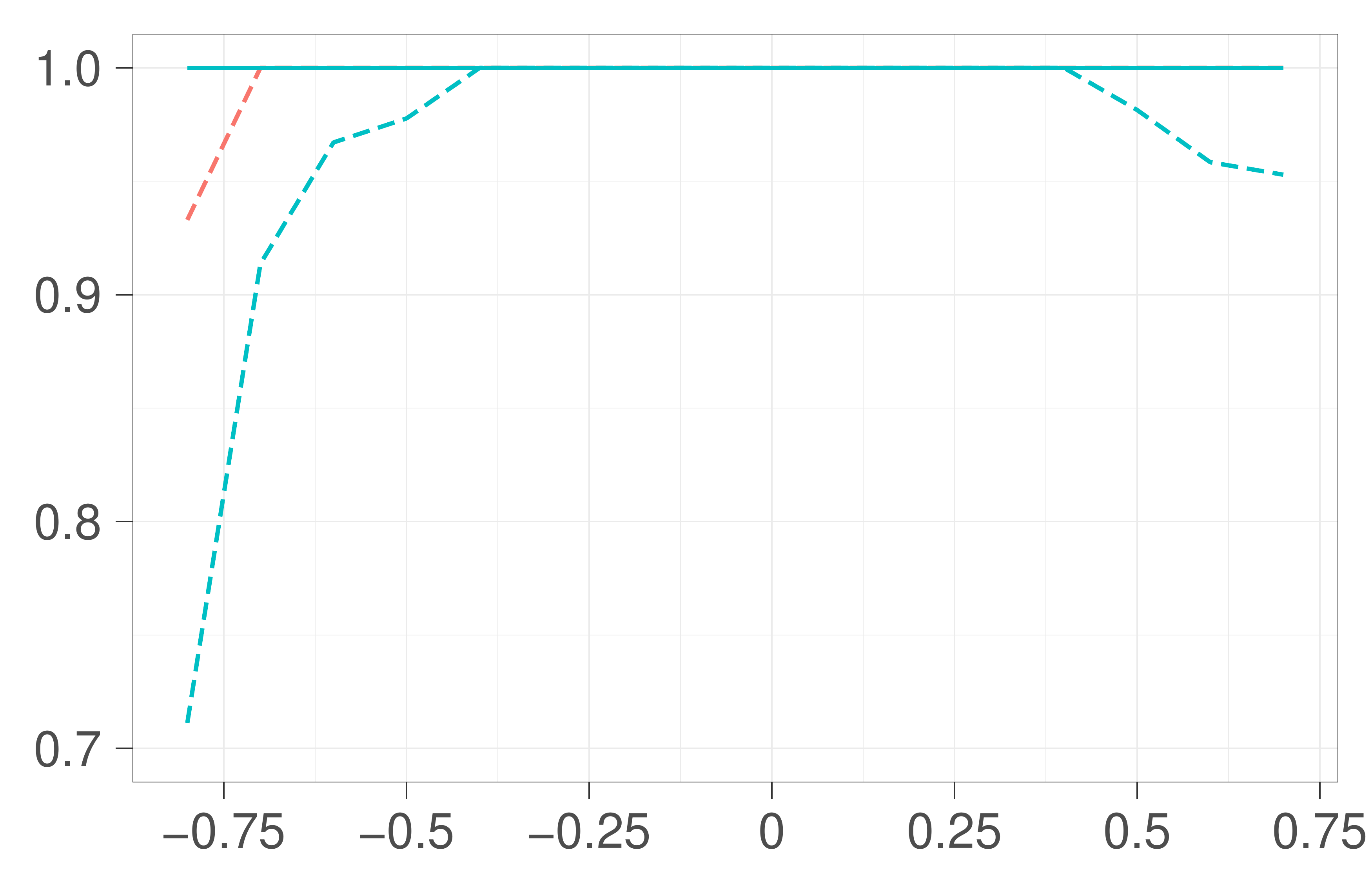}& \\ 
    \rotatebox{90}{\hspace{1.5cm} \small bipartite} &
      \includegraphics[width=.45\textwidth]{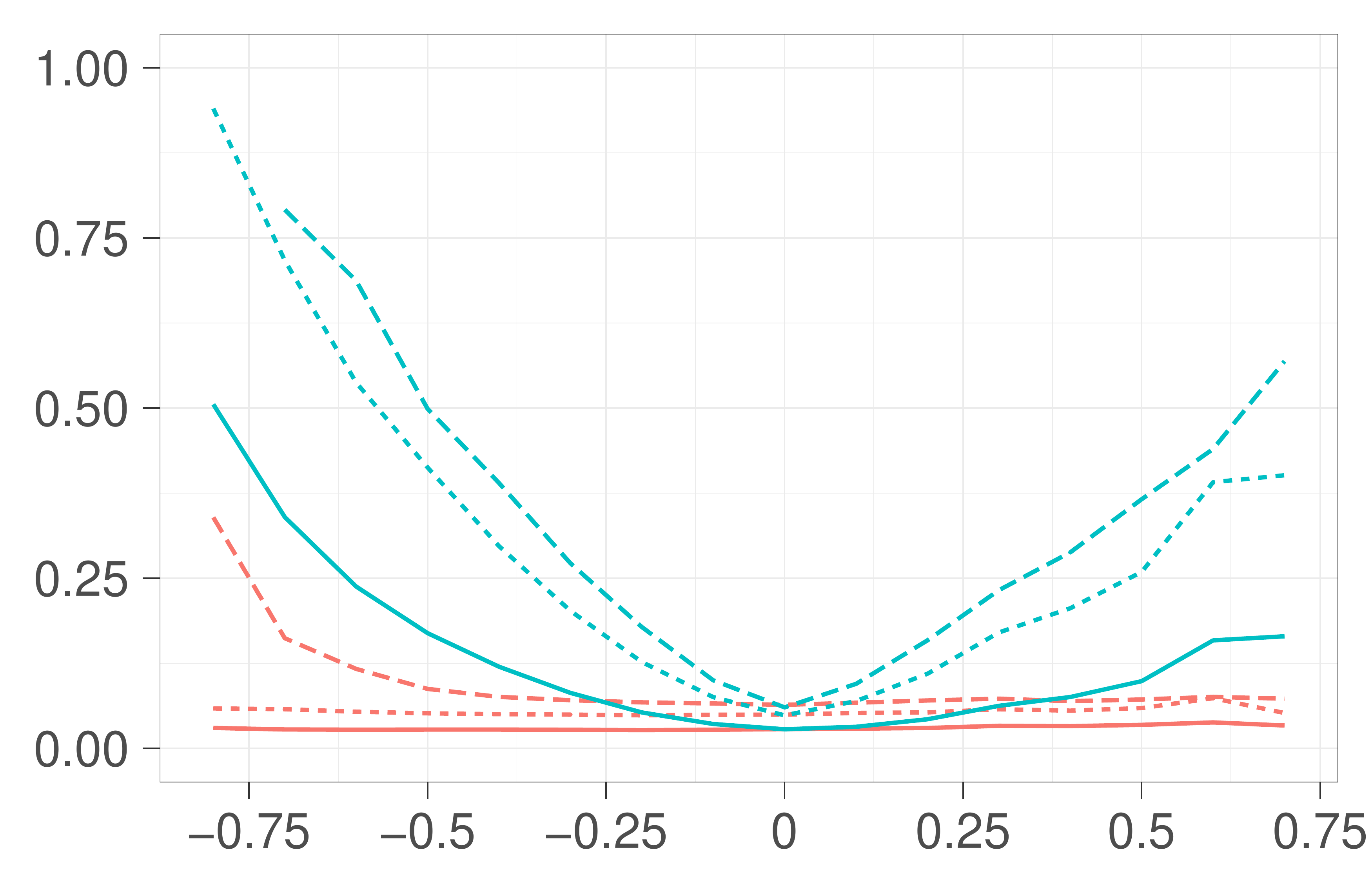}& 
     \includegraphics[width=.45\textwidth]{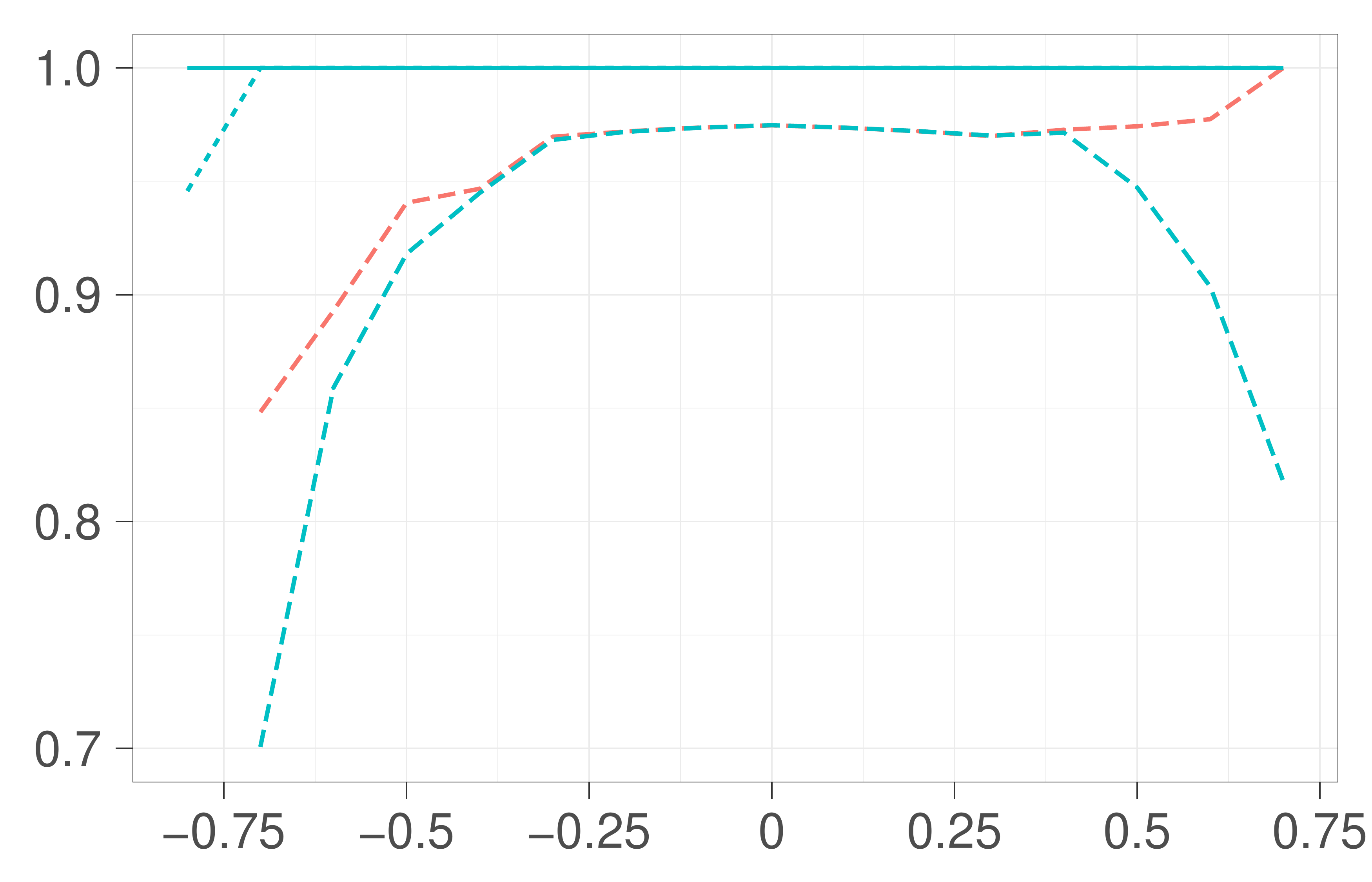} & \includegraphics[width=.08\textwidth]{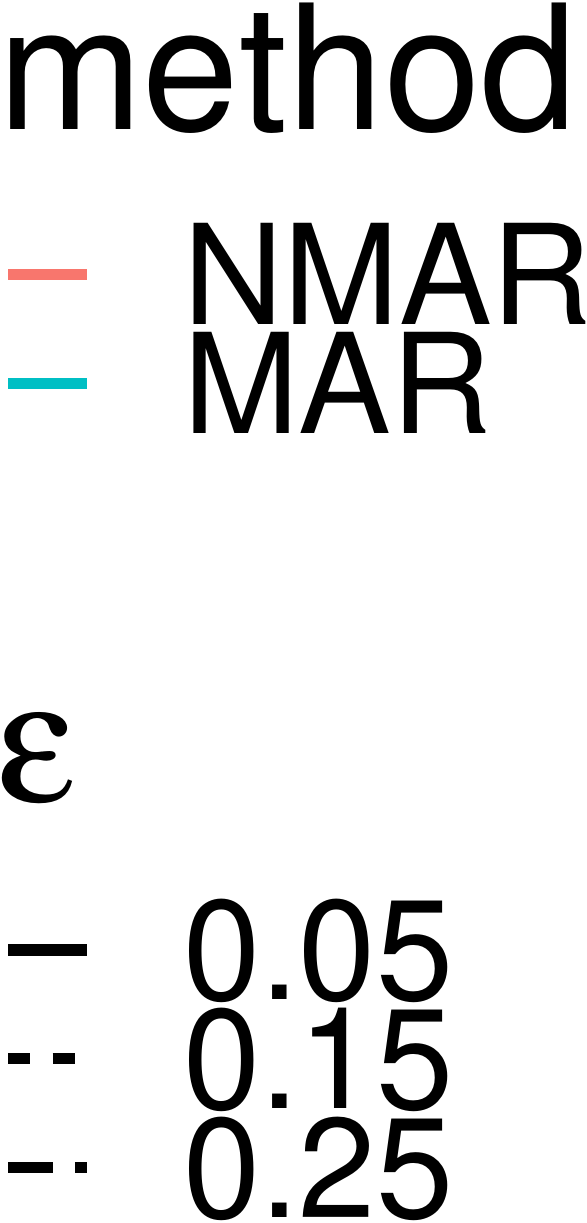}\\
    \rotatebox{90}{\hspace{1.5cm} \small star}&
     \includegraphics[width=.45\textwidth]{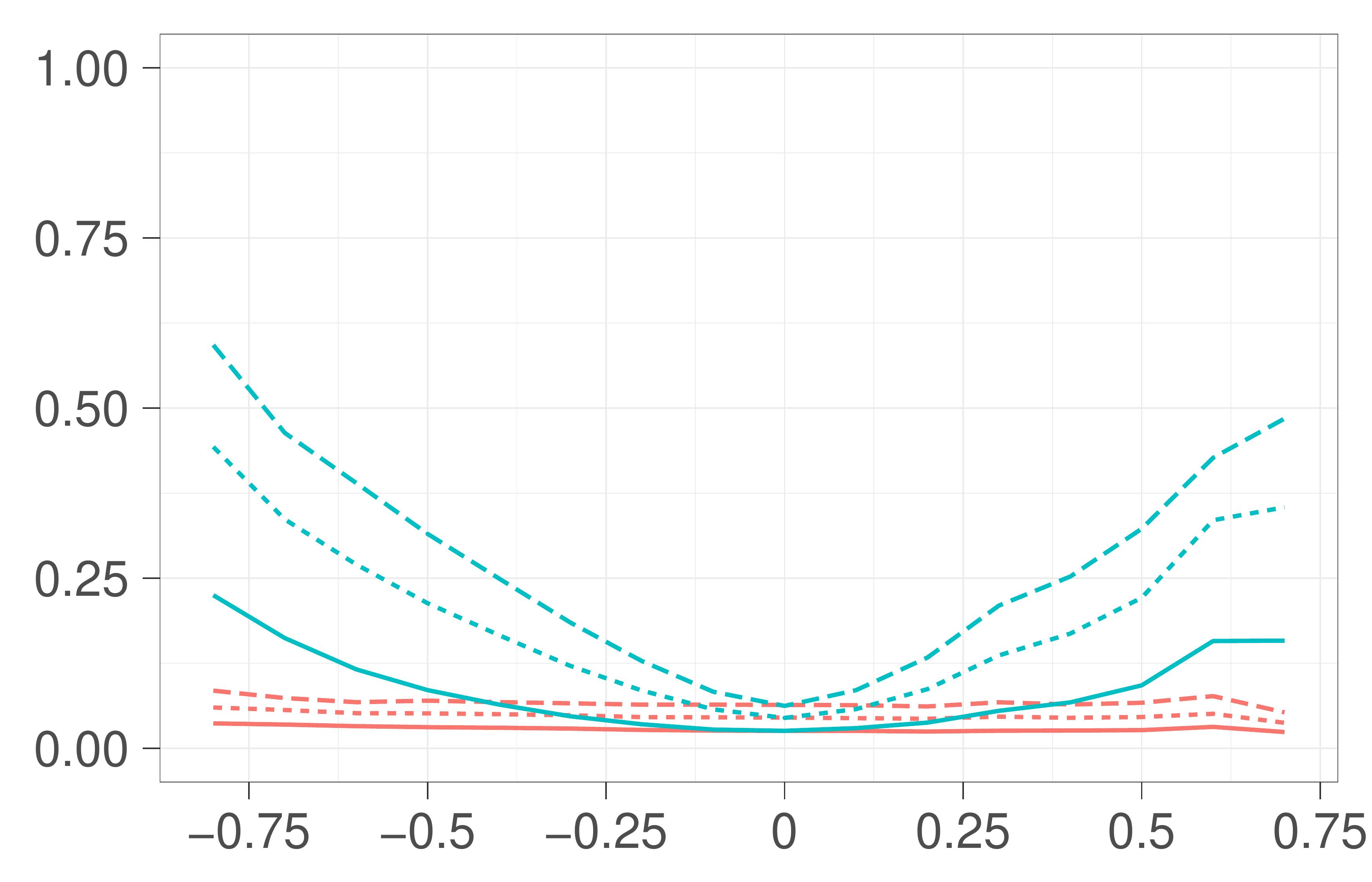}& 
     \includegraphics[width=.45\textwidth]{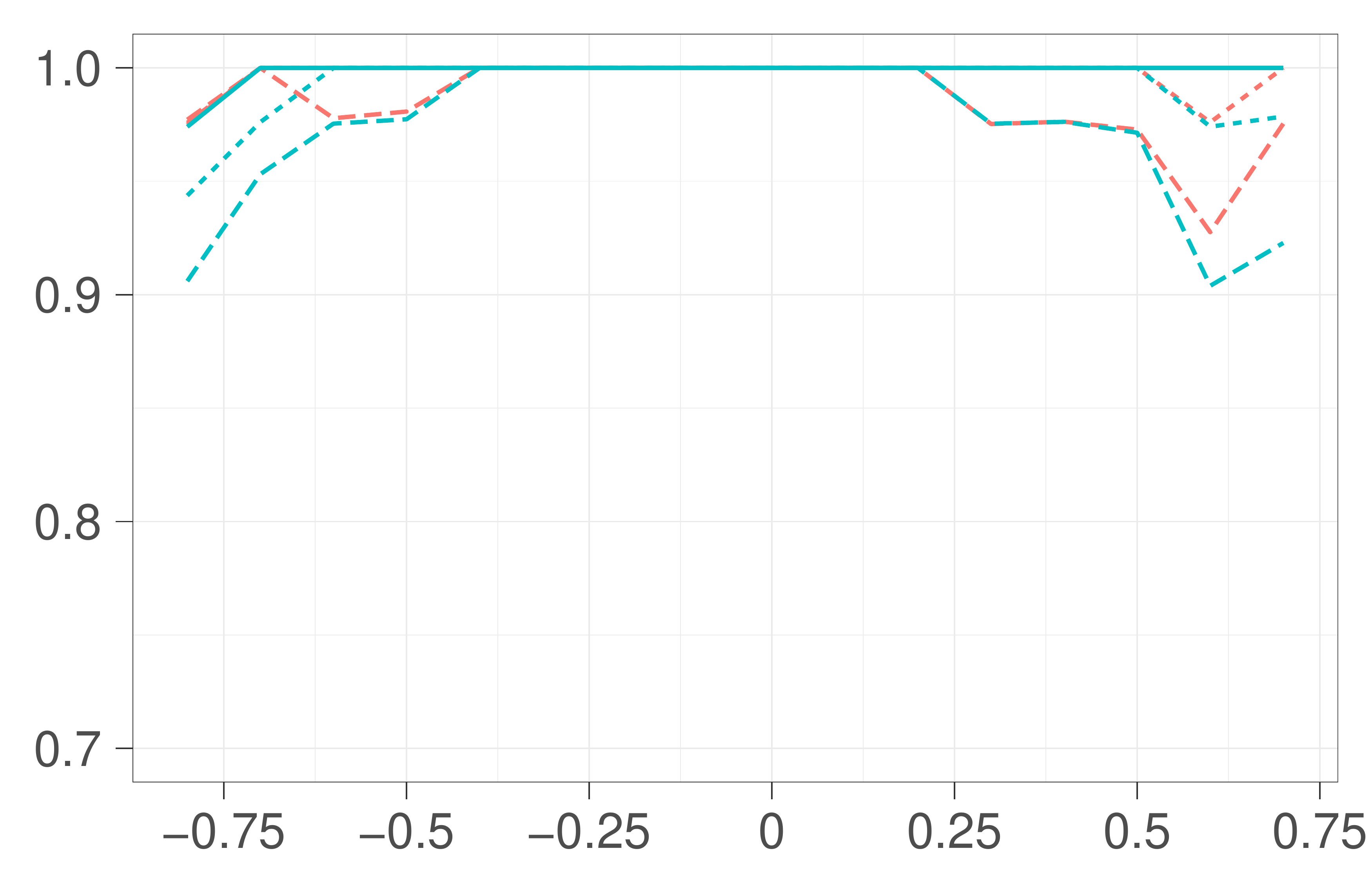} & \\
    & \multicolumn{2}{c}{$\rho_1 - \rho_0$} & \\
  \end{tabular}    
  \caption{Double  standard setting:  estimation  error  of $\pi$  and
    adjusted Rand index averaged over 500 simulations for affiliation,
    bipartite and star topologies.}
  \label{fig:simu_twoStd}
\end{figure}
\begin{figure}[htbp!]
  \centering
  \begin{tabular}{@{}cc@{}}
    $|\hat{\rho}_{0} - \rho_{0} | / | \rho_{0} |$
    & $|\hat{\rho}_{1} - \rho_{1} | / | \rho_{1} |$\\
    \includegraphics[width=.45\textwidth]{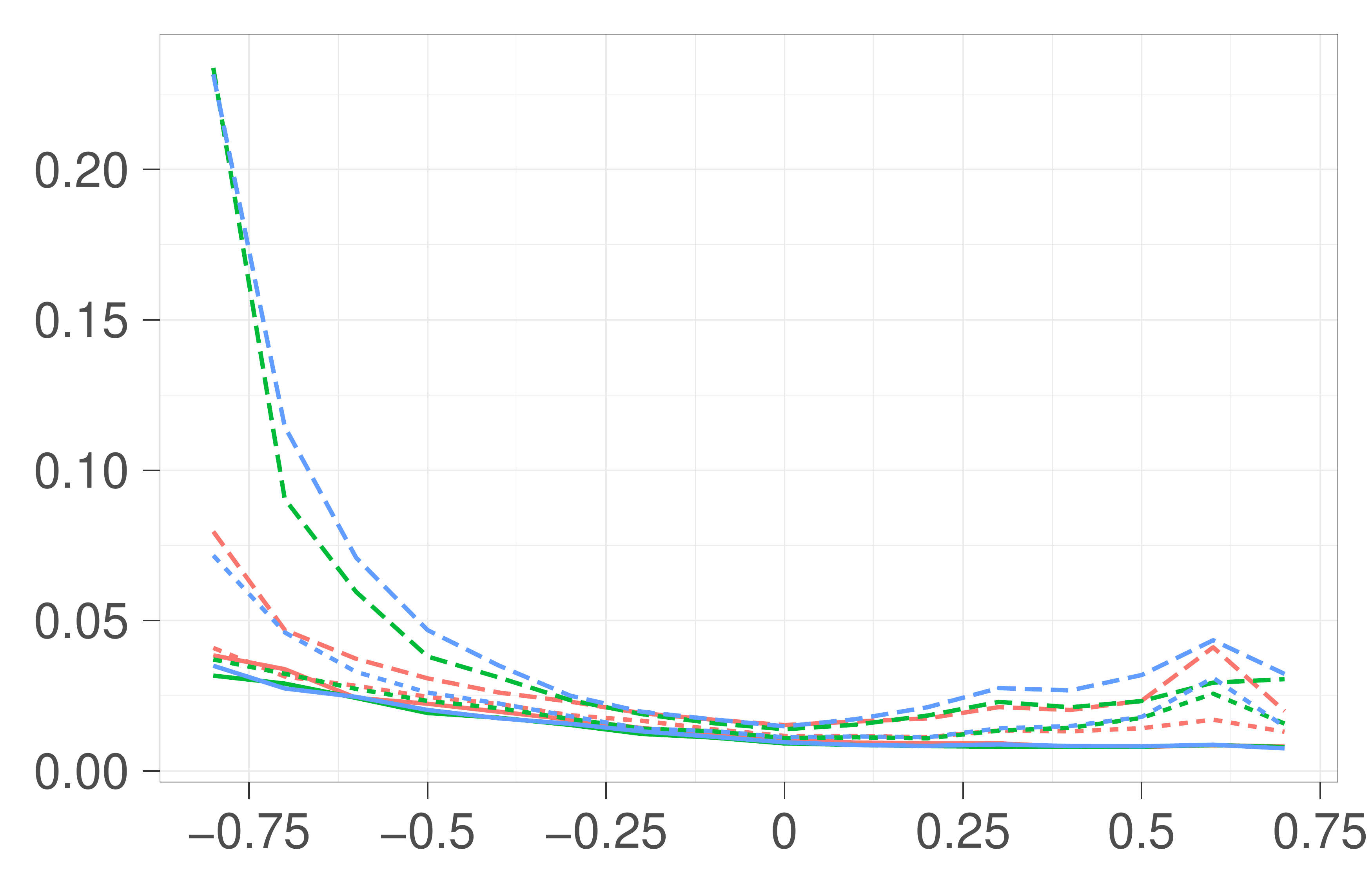}
    & \includegraphics[width=.45\textwidth]{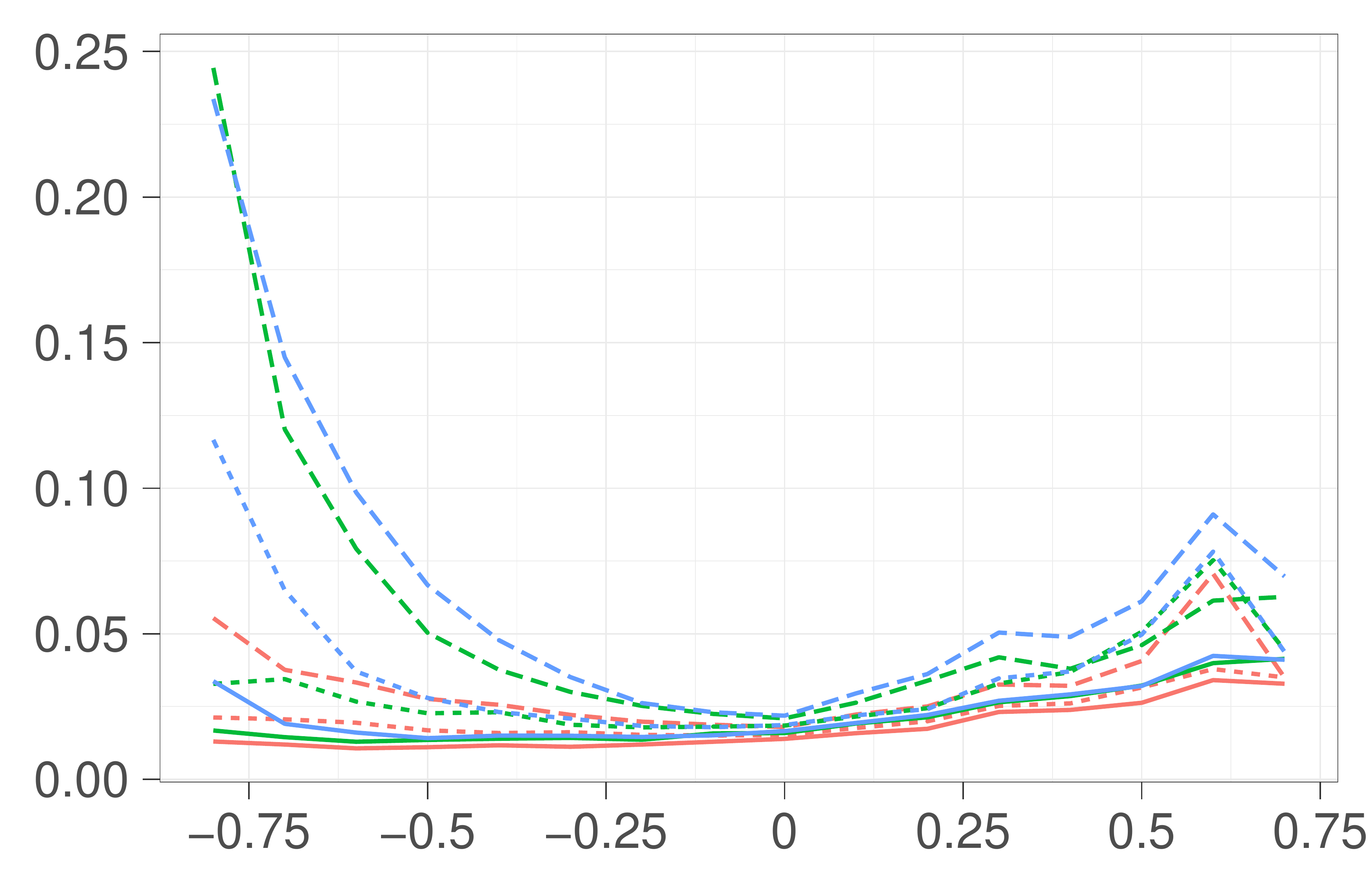} \\ 
    \multicolumn{2}{c}{$\rho_1 - \rho_0$} \\ \\
    \multicolumn{2}{c}{\includegraphics[width=.6\textwidth]{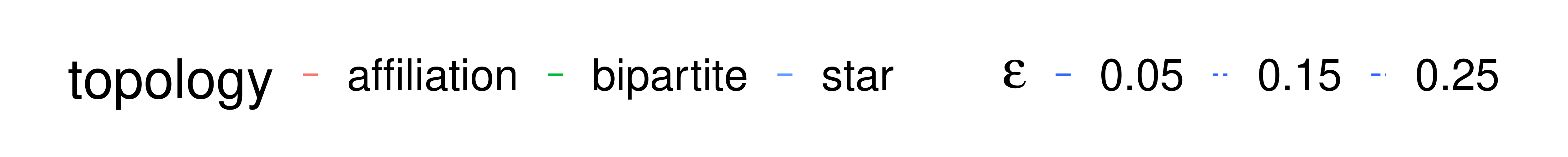}} \\
  \end{tabular}    
  \caption{Double standard setting:  estimation  error  of $\rho_0$  and
    $\rho_1$ averaged over 500 simulations for affiliation,
    bipartite and star topologies.}
  \label{fig:simu_errorRho}
\end{figure}
Figure \ref{fig:simu_errorRho} reports estimation accuracy for the
sampling parameters $\rho_0$ and $\rho_1$. Results show a good ability
of the VEM to estimate these parameters.  As expected, performances
deteriorate for uncontrasted topologies with low sampling rate.

\paragraph{Model selection.}  Simulations are also conducted to study
the performances of ICL.  We compare results for the different
topologies described in Figure \ref{fig:Topologies} for
$\epsilon = 0.05$. Rates of correct answers for selecting the number
of blocks $Q$ under a double standard sampling with different sampling
rates are displayed in Table~\ref{tab:Q_select}.  The ARI is also
provided.  The ICL shows a satisfactory ability to select the true $Q$
even if the selection task obviously needs a larger sampling effort
than the estimation task.  It is worth mentioning that a whole block
may not be sampled, which leads the ICL to select a lower number of
blocks.  In such a case the ARI is a meaningful additional information
to demonstrate that the clustering remains coherent with the true one.

\begin{table}[htbp!]
  \centering
  \begin{tabular}{l|ccc}
   sampling rate & affiliation   & bipartite & star \\ \hline
     (0.154, 0.405] & 0.58/0.96 & 0.46/0.84 & 0.45/0.84 \\
     (0.405, 0.656] & 0.95/0.99 & 0.87/0.98 & 0.90/0.98  \\
     (0.656, 0.908] & 1/1 & 0.99/1 & 0.99/1 \\ \hline
  \end{tabular}
  \caption{Performance of the ICL criterion: 
  rates of correct answers when choosing
    the number of blocks and Adjusted Rand Indexes. Tested configurations are different sampling rates
    and three topologies (affiliation, bipartite and star) under a
    double standard sampling. Each configuration is simulated 500
    times.}
  \label{tab:Q_select}
\end{table}

In Table~\ref{tab:samp_select}, results concern the rates of correct
selections of the sampling design when the two designs in competition
are the random-dyad and the double standard samplings. As expected,
the rate of correct answers increases with the sampling rate.
 
\begin{table}[htbp!]
   \centering
   \begin{tabular}{cc|cccc}
     sampling rate & sampling & affiliation   & bipartite & star \\ \hline
     (0.096, 0.367] & MAR & 0.73 & 0.67 & 0.63 \\
                   & NMAR & 0.72 & 0.75 & 0.75 \\ \hline
     (0.367, 0.638] & MAR & 1 & 1 & 1  \\
                   & NMAR & 0.91 & 0.78 & 0.82  \\\hline
     (0.638, 0.909] & MAR & 1 & 1 & 1 \\
                   & NMAR & 0.91 & 0.8 & 0.95  \\\hline
  \end{tabular}
  \caption{Rates of correct answers of the ICL criterion when choosing
    between Random-dyad sampling (MAR) and double standard sampling
    (NMAR) in each of the 18 configurations. A configuration is the
    combination of a topology (affiliation, bipartite and star), a
    sampling rate and a sampling design. Each configuration is
    simulated 500 times.}
  \label{tab:samp_select}
\end{table}
 

\section{Importance of accouting for missing values in real networks}

\subsection{Seed exchange network in the region of Mount Kenya}
\label{sec:kenya}

In a context of subsistence farming, studies which investigate the
relationships between crop genetic diversity and human cultural
diversity patterns have shown that seed exchanges are embedded in
farmers' social organization.  Data on seed exchanges of sorghum in
the region of Mount Kenya were collected and analyzed in
\cite{labeyrie:2016,labeyrie2014influence}.  The sampling is
node-centered since the exchanges are documented by interviewing
farmers who are asked to declare to whom they gave seeds and from whom
they receive seeds.  Since an interview is time consuming, the
sampling is not exhaustive.  A limited space area was defined where
all the farmers were interviewed.  The network is thus collected with
missing dyads since information on the potential links between two
farmers who were cited but not interviewed is missing. With the
courtesy of Vanesse Labeyrie, we analyzed the Mount Kenya seed
exchange network involving $568$ farmers among which $155$ were
interviewed. Although other farmers in this region might be connected
to non-interviewed farmers, we focus on this closed network of $568$
nodes.

Since we only know that the sampling is node-centered, we fit SBM
under the three node-centered sampling designs presented in Section
\ref{subsec:missingdata} (star (MAR), class and star degree sampling).
The ICL criterion is minimal for $10$ blocks under the star degree
sampling and for $11$ blocks under the class degree sampling.  The
clusterings between the SBMs obtained with either class or star degree
sampling remain close from each other (ARI: $0.6$) and both unravel a
strong community structure. The model selected by ICL for MAR sampling
is composed by $11$ blocks.  The ARIs between MAR clustering and the
two other clusterings are lower (around $0.4$).  Finally, note that
interviewed and non-interviewed farmers are mixed up in the blocks of
the three selected models.  The ICL criteria computed for the three
sampling designs are a slightly in favor of the MAR sampling.

On top of network data, categorical variables are available for
discriminating the farmers such as the ntora\footnote{The ntora is a
  small village or a group of neighborhoods} they belong to ($10$ main
ntoras plus $1$ grouping all the others) and the dialect they speak
($4$ dialects).  In Figure \ref{fig:aristoradialect}, we compute ARIs
between the ntoras (left panel), the dialects (right panel) and the
clusterings obtained with the SBM under the three node-centered
sampling designs for a varying number of blocks.  Even though the ARIs
remain low, the clusterings from class or star degree sampling seem to
catch a non negligible fraction of the social organization, larger
than the one caught by the clustering from the MAR sampling. These two
categorical variables, reflecting some aspects of the social
organization, could partially explain the structure of the exchange
network.

\begin{figure}[htbp!]
  \centering
  \begin{tabular}{ccc}
     \rotatebox{90}{\hspace{1.7cm} \small ARI} & \includegraphics[width=.45\textwidth]{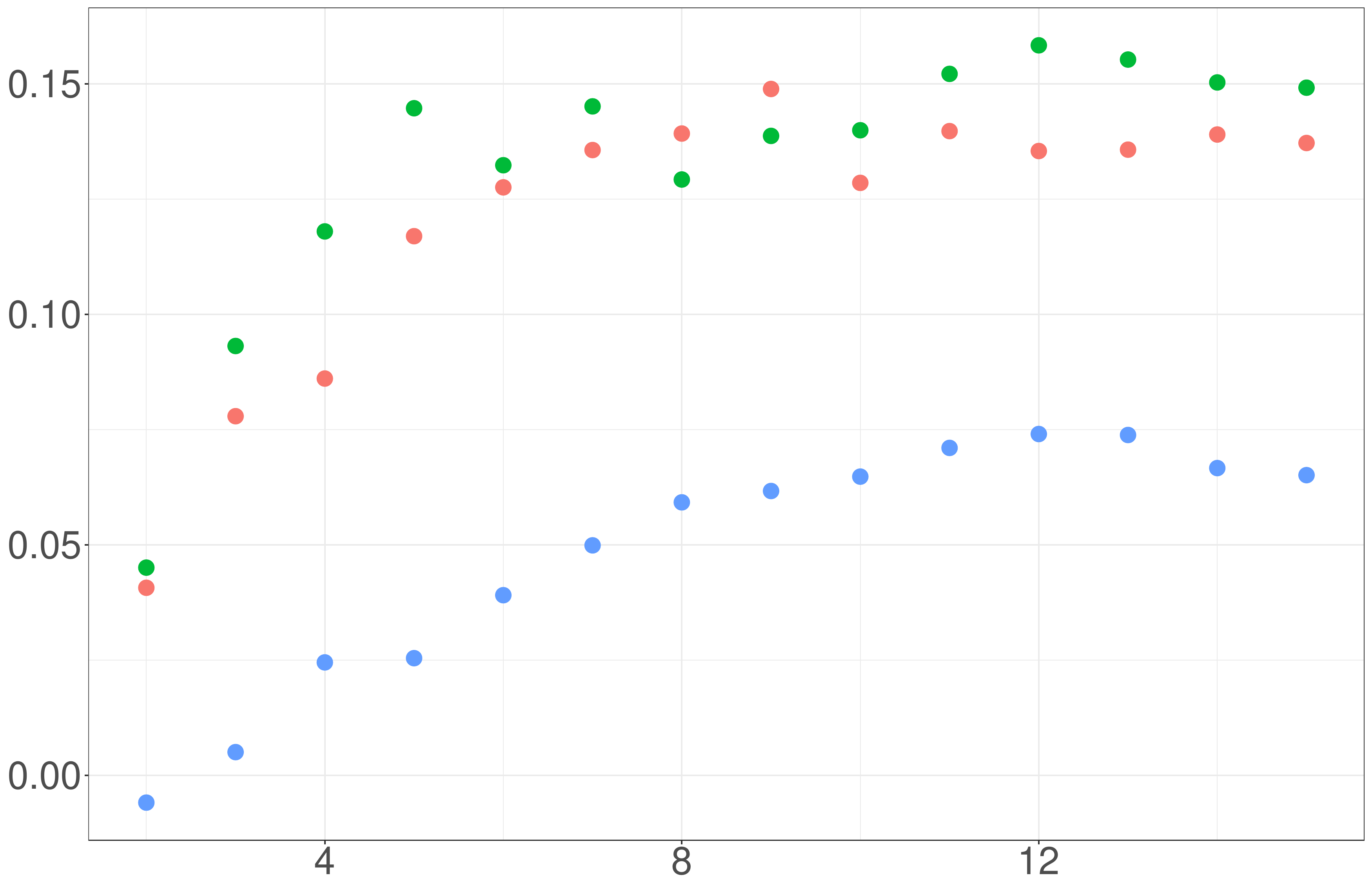}
    & \includegraphics[width=.45\textwidth]{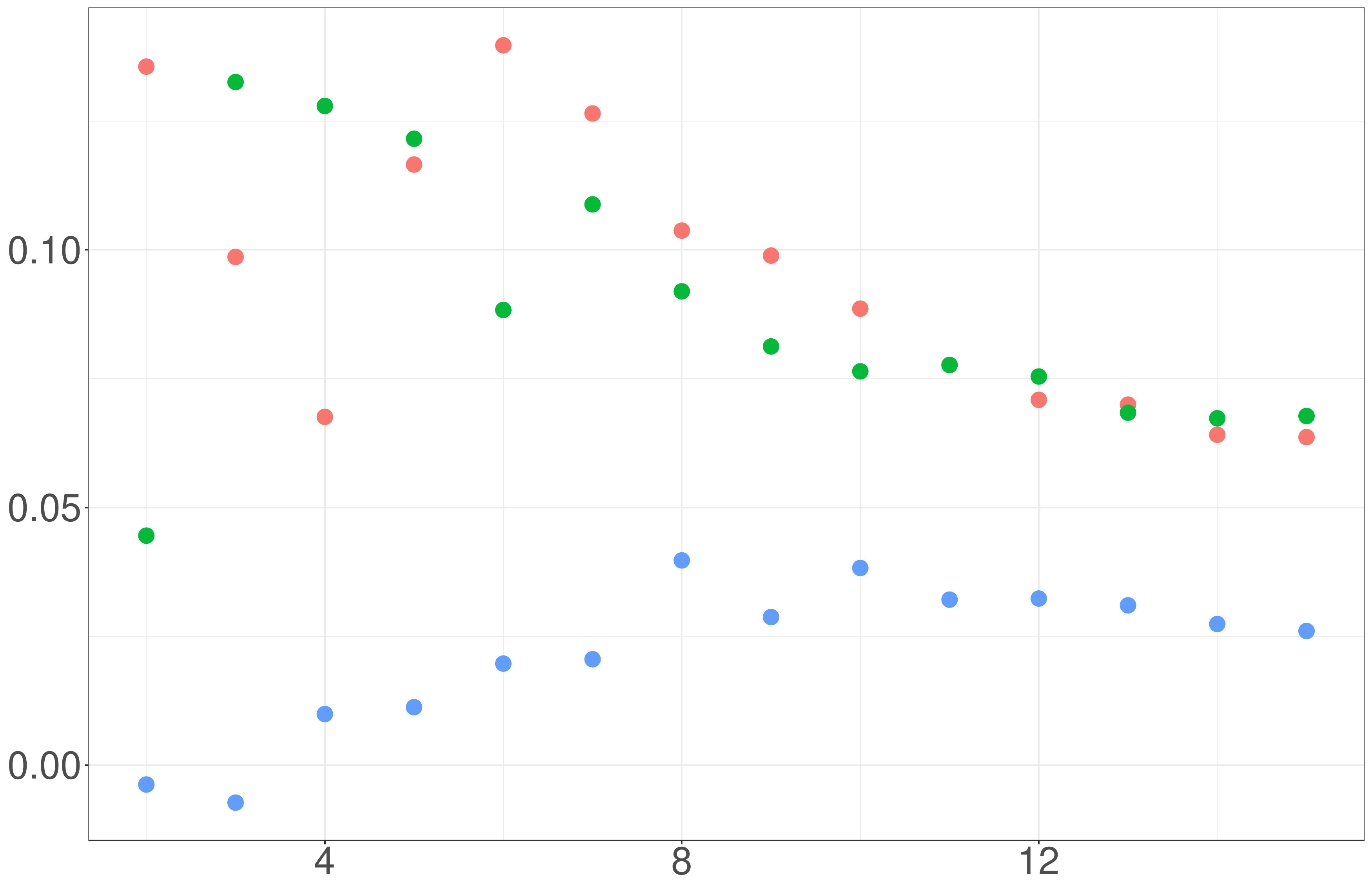} \\ 
    \multicolumn{3}{c}{number of blocks} \\ 
    \multicolumn{3}{c}{\includegraphics[width=.4\textwidth]{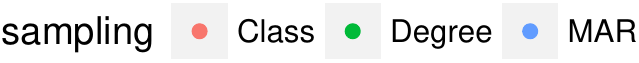}} \\
  \end{tabular}    
  \caption{ARIs computed between the clusterings given by an SBM under class, star degree and MAR samplings with a varying number of blocks $Q$ and 
  ntora of farmers (left-hand-side) or dialect spoken by farmers (right-hand-side)}
  \label{fig:aristoradialect}
\end{figure}

\subsection{ER (ESR1) Protein-Protein Interaction network in breast
  cancer}
\label{sec:er1}


Estrogen receptor 1 (ESR1) is a gene that encodes an estrogen receptor
protein (ER), a central actor in breast cancer. Uncovering its
relations with other proteins is essential for a better understanding
of the disease.  To this end, various bioinformatics tools are
available to centralize knowledge about possible relations between
proteins into networks known as \textit{Protein-Protein Interaction}
(PPI) networks.  The platform \texttt{string} \citep{string:2015}
accessible via \url{http://www.string-db.org} is one of the most
popular tools for this task.  Given a set of one (or several) initial
protein(s) provided by the user, it is possible to recover a valued
network between all proteins connected to the initial set.  The value
of an edge in this network corresponds to a score obtained by
aggregating different types of knowledge (wet-lab experiments,
textmining, co-expression data, etc\dots), reflecting a level of
confidence. Thus, it is possible for a given protein -- we choose ER
here -- to recover the PPI network between all proteins involved.  Our
ambition is to rely on a SBM with missing data to finely analyze such
networks: we rather describe a dyad as missing (thus not choosing
between 0 or 1) if its level of confidence is too low.

The PPI network in the neighborhood of ER is composed by 741 proteins
connected by edges with values in $(0,1]$. We remove ER from this set
of proteins, as well as the zinc finger protein 44. Indeed, they were
both connected to most of the other proteins and would thus only blur
the underlying clustering structure.  We denote $\omega_{ij}$ the
weight associated with dyad $(i,j)$. By means of a tuning parameter
$\gamma$ reflecting the level of confidence, the adjacency matrix is
defined as follows:
\begin{equation}
  \label{eq:adjacency_string}
  \mathbf{A}^\gamma = \left(A^\gamma\right)_{ij} = \left\{
    \begin{array}{ll}
      1 & \text{if } \omega_{ij} > 1-\gamma, \\
      \texttt{NA} & \text{if } \gamma \leq \omega_{ij} \leq 1-\gamma, \\
      0 & \text{if } \omega_{ij} < \gamma. \\
    \end{array}
  \right.
\end{equation}
In order to analyze the ER-centered network,
Algorithm~\ref{algo:vem:mar} (random-dyad MAR sampling) and
Algorithm~\ref{algo:vem:nmar} (double-standard NMAR sampling) were
applied on $\mathbf{A}^{\gamma}$ for $\gamma$ varying in
$\{.15, .25, .35\}$, hence taking the uncertainties on the missing
dyads into account with various thresholds.  The ICL criterion in
Figure~\ref{fig:graph_ICL_er} systematically chooses the NMAR modeling
against the MAR modeling, whatever the value of $\gamma$.
\begin{figure}[htbp!]
  \begin{tabular}{@{}l@{\hspace{.1ex}}c@{}c@{}c@{}}
    & \small $\gamma = .15$ & \small $\gamma = .25$ & \small $\gamma = .35$ \\  
    \rotatebox{90}{\hspace{2.5cm} \small ICL}
    & \includegraphics[width=.325\textwidth]{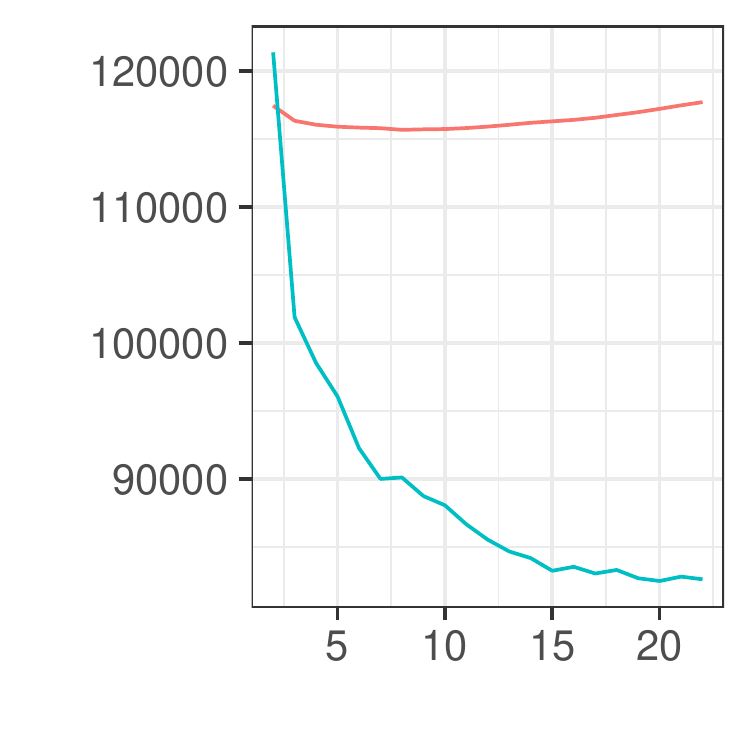}
    & \includegraphics[width=.325\textwidth]{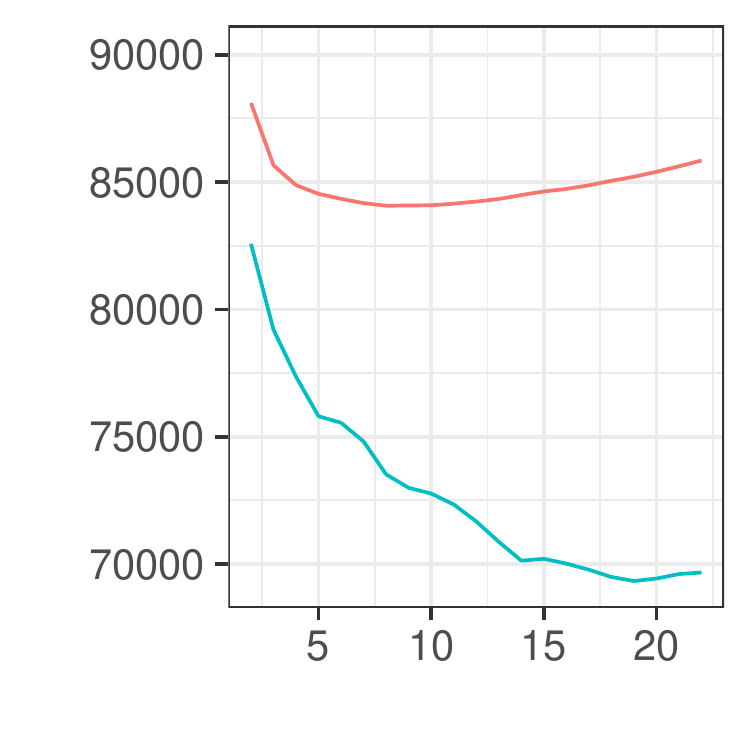}
    & \includegraphics[width=.325\textwidth]{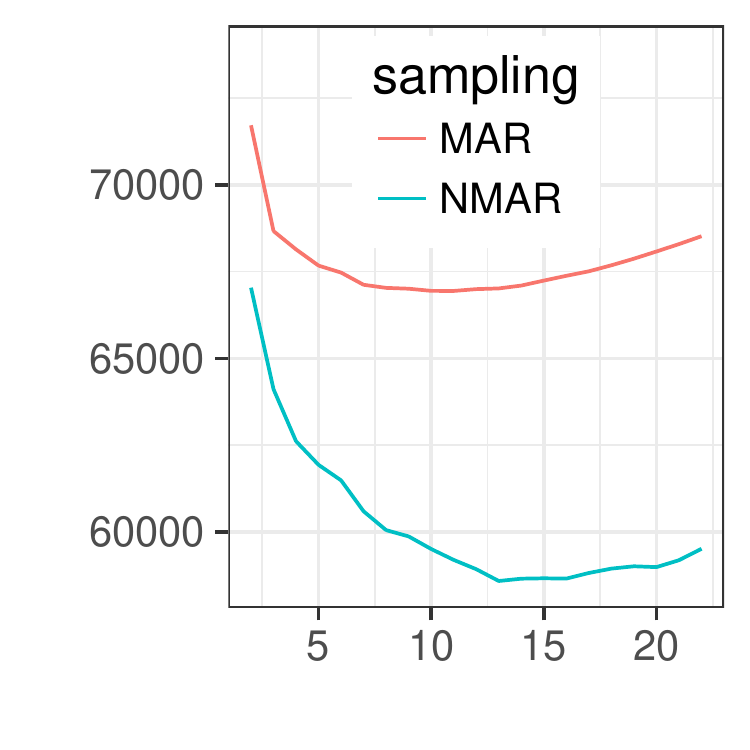} \\
    & & number of blocks & \\
  \end{tabular}
  \caption{ICL criteria for
      SBMs with random-dyad MAR sampling and double-standard
      NMAR sampling in the thresholded ER network.}
  \label{fig:graph_ICL_er}
\end{figure}

We study the best MAR and NMAR models associated with $\gamma = 0.35$,
which value exhibits a clearer choice of the ICL than for
$\gamma = \{.15, .25\}$ for both MAR and NMAR modelings. The two
corresponding SBMs have 11 clusters for MAR sampling and 13 clusters
for NMAR sampling.  The ARI between the two clusterings is around
$0.39$: this is mainly due to a large block in the random-dyad MAR
clustering which contains much more nodes than any of the blocks in
the NMAR clustering. The latter dispatches many of these nodes in four
blocks (see the supplementary materials for a more detailed exposition
of results).  To prove that this finest clustering of the nodes is
more relevant from the biological point of view, we propose a
validation based on external biological knowledge. To this end, we
rely on the Gene Ontology (GO) annotation \citep{ashburner2000gene}
which provides a DAG of ontologies to which genes are annotated if the
proteins encoded by these genes are involved in a known biological
process.  Here, we use GO to perform enrichment analysis (that is to
say identifying classes of genes that are over-represented in a large
set of genes, via a simple hypergeometric test) on genes corresponding
to the proteins present in the large block for MAR, and the
corresponding four blocks for NMAR. Interestingly, at a significance
level of $1\%$, we find a single significant biological process for
MAR modeling while 13 are found significant in the NMAR case.  We
check that it is not due to a simple threshold effect by looking at
the ranks of the $p$-values of the 13 NMAR significant processes in
the 100 first most significant terms found in the MAR model: only 5 of
the NMAR processes are found, with high ranks (24, 33, 39, 56 and 77)
far from the smallest MAR $p$-values.
 	

\section{Conclusion}

This paper shows how to deal with missing data on dyads in the SBM. We
study MAR and NMAR sampling designs motivated by network data and
propose variational approaches to perform inference under this series
of designs, accompanied with model selection criteria.  Relevance of
the method is illustrated on numerical experiments both on simulated
and real-world networks. An \texttt{R}-package is available at
\ifblinded
\url{URL_blinded}
\else
\url{https://github.com/jchiquet/missSBM}.
\fi.

This work focuses on undirected binary networks. However, it can be
adapted to other SBMs, in particular those developed in
\cite{mariadassou2010} for (un)directed valued networks with a
distribution of weights belonging to the exponential family. It could
also be adapted to the degree-corrected SBM \citep{Karrer2011}, where
the sampling design would depend on the degree correction
parameters. This should lead to a design close to the star degree
sampling.  In future works, we plan to investigate the consistency of
the variational estimators of SBM under missing data conditions,
looking for similar results as the ones obtained in
\cite{bickel2013asymptotic} for fully observed networks.  Another path
of research is to consider missing data where we cannot distinguish
between a missing dyad and the absence of an edge like in
\citet{Priebe2015,Balachandran2017}.


\paragraph{Acknowledgment.} 
\ifblinded
\textit{text blinded to preserve anonymity.}
\else
The authors thank Sophie Donnet (INRA-MIA, AgroParisTech) and Mahendra
Mariadassou (INRA-MaIAGE, Jouy-en-Josas) for their helpful remarks
and suggestions.  We also thank all members of MIRES for fruitful
discussions on network sampling designs and for providing the original
problems from social science. In particular, we thank Vanesse Labeyrie
(CIRAD-Green) for sharing the seed exchange data and for related
discussions on the analysis.

This work is supported by the INRA MetaProgram « GloFoods » through
the project « SEARS » and by two public grants overseen by the French
National research Agency (ANR): first as part of the « Investissement
d’Avenir » program, through the « IDI 2017 » project funded by the
IDEX Paris-Saclay, ANR-11-IDEX-0003-02, and second by the « EcoNet »
project.
\fi

\bibliography{sbmsampling}

\begin{thebibliography}{39}
\providecommand{\natexlab}[1]{#1}
\providecommand{\url}[1]{\texttt{#1}}
\expandafter\ifx\csname urlstyle\endcsname\relax
  \providecommand{\doi}[1]{doi: #1}\else
  \providecommand{\doi}{doi: \begingroup \urlstyle{rm}\Url}\fi

\bibitem[Aicher et~al.(2014)Aicher, Jacobs, and Clauset]{Aicher2014}
C.~Aicher, A.~Z. Jacobs, and A.~Clauset.
\newblock Learning latent block structure in weighted networks.
\newblock \emph{J. Compl. Net.}, 3.2:\penalty0 221--248, 2014.

\bibitem[Airoldi et~al.(2008)Airoldi, Blei, Fienberg, and
  Xing]{airoldi2008mixed}
E.~M. Airoldi, D.~M. Blei, S.~E. Fienberg, and E.~P. Xing.
\newblock Mixed membership stochastic blockmodels.
\newblock \emph{J. Mach. Learn. Res.}, 9\penalty0 (Sep):\penalty0 1981--2014,
  2008.

\bibitem[Ashburner et~al.(2000)Ashburner, Ball, Blake, Botstein, Butler,
  Cherry, Davis, Dolinski, Dwight, Eppig, et~al.]{ashburner2000gene}
M.~Ashburner, C.~A. Ball, J.~A. Blake, D.~Botstein, H.~Butler, J.~M. Cherry,
  A.~P. Davis, K.~Dolinski, S.~S. Dwight, J.~T. Eppig, et~al.
\newblock Gene ontology: tool for the unification of biology.
\newblock \emph{Nat. Genet.}, 25\penalty0 (1):\penalty0 25, 2000.

\bibitem[Balachandran et~al.(2017)Balachandran, Kolaczyk, and
  Viles]{Balachandran2017}
P.~Balachandran, E.~D. Kolaczyk, and W.~D. Viles.
\newblock On the propagation of low-rate measurement error to subgraph counts
  in large networks.
\newblock \emph{Journal of Machine Learning Research}, 18\penalty0
  (61):\penalty0 1--33, 2017.

\bibitem[Barbillon et~al.(2015)Barbillon, Donnet, Lazega, and
  Bar-Hen]{Barbillon2015}
P.~Barbillon, S.~Donnet, E.~Lazega, and A.~Bar-Hen.
\newblock Stochastic block models for multiplex networks: an application to
  networks of researchers.
\newblock \emph{J. R. Stat. Soc. C-Appl.}, 2015.

\bibitem[Bickel et~al.(2013)Bickel, Choi, Chang, Zhang,
  et~al.]{bickel2013asymptotic}
P.~Bickel, D.~Choi, X.~Chang, H.~Zhang, et~al.
\newblock Asymptotic normality of maximum likelihood and its variational
  approximation for stochastic blockmodels.
\newblock \emph{Ann. Stat.}, 41\penalty0 (4):\penalty0 1922--1943, 2013.

\bibitem[Biernacki et~al.(2000)Biernacki, Celeux, and Govaert]{Biernacki2000}
C.~Biernacki, G.~Celeux, and G.~Govaert.
\newblock Assessing a mixture model for clustering with the integrated
  completed likelihood.
\newblock \emph{IEEE Trans. Pattern Anal. Mach. Intell.}, 22\penalty0
  (7):\penalty0 719--725, 2000.

\bibitem[Celisse et~al.(2012)Celisse, Daudin, Pierre,
  et~al.]{celisse2012consistency}
A.~Celisse, J.-J. Daudin, L.~Pierre, et~al.
\newblock Consistency of maximum-likelihood and variational estimators in the
  stochastic block model.
\newblock \emph{Electron. J. Stat.}, 6:\penalty0 1847--1899, 2012.

\bibitem[Chatterjee(2015)]{Chatterjee2015}
S.~Chatterjee.
\newblock Matrix estimation by universal singular value thresholding.
\newblock \emph{The Annals of Statistics}, 43\penalty0 (1):\penalty0 177--214,
  2015.

\bibitem[Daudin et~al.(2008)Daudin, Picard, and Robin]{daudin2008mixture}
J.-J. Daudin, F.~Picard, and S.~Robin.
\newblock A mixture model for random graphs.
\newblock \emph{Stat. comp.}, 18\penalty0 (2):\penalty0 173--183, 2008.

\bibitem[Davenport et~al.(2014)Davenport, Plan, van~den Berg, and
  Wootters]{Davenport2014}
M.~A. Davenport, Y.~Plan, E.~van~den Berg, and M.~Wootters.
\newblock 1-bit matrix completion.
\newblock \emph{Information and Inference: A Journal of the IMA}, 3\penalty0
  (3):\penalty0 189--223, 2014.

\bibitem[Frank and Harary(1982)]{frank1982cluster}
O.~Frank and F.~Harary.
\newblock Cluster inference by using transitivity indices in empirical graphs.
\newblock \emph{J. Am. Stat. Soc.}, 77\penalty0 (380):\penalty0 835--840, 1982.

\bibitem[Goldenberg et~al.(2010)Goldenberg, Zheng, Fienberg, Airoldi,
  et~al.]{goldenberg2010survey}
A.~Goldenberg, A.~X. Zheng, S.~E. Fienberg, E.~M. Airoldi, et~al.
\newblock A survey of statistical network models.
\newblock \emph{Foundations and Trends{\textregistered} in Machine Learning},
  2\penalty0 (2):\penalty0 129--233, 2010.

\bibitem[Handcock and Gile(2010)]{Handcock2010}
M.~S. Handcock and K.~J. Gile.
\newblock Modeling social networks from sampled data.
\newblock \emph{The Annals of Applied Statistics}, 4\penalty0 (1):\penalty0
  5--25, 2010.

\bibitem[Holland et~al.(1983)Holland, Laskey, and
  Leinhardt]{holland1983stochastic}
P.~W. Holland, K.~B. Laskey, and S.~Leinhardt.
\newblock Stochastic blockmodels: First steps.
\newblock \emph{Social networks}, 5\penalty0 (2):\penalty0 109--137, 1983.

\bibitem[Jordan et~al.(1998)Jordan, Ghahramani, Jaakkola, and
  Saul]{jordan1998introduction}
M.~I. Jordan, Z.~Ghahramani, T.~S. Jaakkola, and L.~K. Saul.
\newblock An introduction to variational methods for graphical models.
\newblock In \emph{Learning in graphical models}, pages 105--161. Springer,
  1998.

\bibitem[Karrer and Newman(2011)]{Karrer2011}
B.~Karrer and M.~E.~J. Newman.
\newblock Stochastic blockmodels and community structure in networks.
\newblock \emph{Phys. Rev. E}, 83:\penalty0 016107, Jan 2011.

\bibitem[Kolaczyk(2009)]{Kolaczyk2009}
E.~D. Kolaczyk.
\newblock \emph{Statistical analysis of network data, methods and models}.
\newblock Springer, 2009.

\bibitem[Labeyrie et~al.(2014)Labeyrie, Deu, Barnaud, Calatayud, Buiron,
  Wambugu, Manel, Glaszmann, and Leclerc]{labeyrie2014influence}
V.~Labeyrie, M.~Deu, A.~Barnaud, C.~Calatayud, M.~Buiron, P.~Wambugu, S.~Manel,
  J.-C. Glaszmann, and C.~Leclerc.
\newblock Influence of ethnolinguistic diversity on the sorghum genetic
  patterns in subsistence farming systems in eastern kenya.
\newblock \emph{PLoS One}, 9\penalty0 (3):\penalty0 e92178, 2014.

\bibitem[Labeyrie et~al.(2016)Labeyrie, Thomas, Muthamia, and
  Leclerc]{labeyrie:2016}
V.~Labeyrie, M.~Thomas, Z.~K. Muthamia, and C.~Leclerc.
\newblock Seed exchange networks, ethnicity, and sorghum diversity.
\newblock \emph{P. Natl. Acad. Sci.}, 113\penalty0 (1):\penalty0 98--103, 2016.

\bibitem[Latouche et~al.(2011)Latouche, Birmel{\'e}, and
  Ambroise]{latouche2011overlapping}
P.~Latouche, {\'E}.~Birmel{\'e}, and C.~Ambroise.
\newblock Overlapping stochastic block models with application to the french
  political blogosphere.
\newblock \emph{Ann. Appl. Stat.}, pages 309--336, 2011.

\bibitem[Latouche et~al.(2012)Latouche, Birmel\'e, and
  Ambroise]{latouche2012variational}
P.~Latouche, {\'E}.~Birmel\'e, and C.~Ambroise.
\newblock Variational bayesian inference and complexity control for stochastic
  block models.
\newblock \emph{Stat. Modelling}, 12\penalty0 (1):\penalty0 93--115, 2012.

\bibitem[Latouche et~al.(2017)Latouche, Robin, and
  Ouadah]{LatoucheRobinOuadah2017}
P.~Latouche, S.~Robin, and S.~Ouadah.
\newblock {Goodness of fit of logistic models for random graphs}.
\newblock Technical report, 2017.
\newblock URL \url{https://arxiv.org/abs/1508.00286}.

\bibitem[Little and Rubin(2014)]{little2014statistical}
R.~J. Little and D.~B. Rubin.
\newblock \emph{Statistical analysis with missing data}.
\newblock John Wiley \& Sons, 2014.

\bibitem[Mariadassou et~al.(2010)Mariadassou, Robin, and
  Vacher]{mariadassou2010}
M.~Mariadassou, S.~Robin, and C.~Vacher.
\newblock Uncovering latent structure in valued graphs: A variational approach.
\newblock \emph{Ann. Appl. Stat.}, 4\penalty0 (2):\penalty0 715--742, 06 2010.

\bibitem[Matias and Miele(2016)]{matias2016statistical}
C.~Matias and V.~Miele.
\newblock Statistical clustering of temporal networks through a dynamic
  stochastic block model.
\newblock \emph{J. R. Stat. Soc. B-Met.}, 2016.

\bibitem[Matias and Robin(2014)]{matias2014modeling}
C.~Matias and S.~Robin.
\newblock Modeling heterogeneity in random graphs through latent space models:
  a selective review.
\newblock \emph{ESAIM Proc. Sur.}, 47:\penalty0 55--74, 2014.

\bibitem[Molenberghs et~al.(2008)Molenberghs, Beunckens, Sotto, and
  Kenward]{Molenberghs2008}
G.~Molenberghs, C.~Beunckens, C.~Sotto, and G.~M. Kenward.
\newblock Every missing not at random model has got a missing at random
  counterpart with equal fit.
\newblock \emph{J. R. Stat. Soc. B-Met.}, 2008.

\bibitem[Nowicki and Snijders(2001)]{Nowicki2001}
K.~Nowicki and T.~A.~B. Snijders.
\newblock Estimation and prediction for stochastic blockstructures.
\newblock \emph{J. Am. Stat. Soc.}, 96\penalty0 (455):\penalty0 1077--1087,
  September 2001.

\bibitem[Priebe et~al.(2015)Priebe, Sussman, Tang, and Vogelstein]{Priebe2015}
C.~E. Priebe, D.~L. Sussman, M.~Tang, and J.~T. Vogelstein.
\newblock Statistical inference on errorfully observed graphs.
\newblock \emph{Journal of Computational and Graphical Statistics}, 24\penalty0
  (4):\penalty0 930--953, 2015.

\bibitem[Rand(1971)]{rand1971}
W.~M. Rand.
\newblock Objective criteria for the evaluation of clustering methods.
\newblock \emph{J. Am. Stat. Soc.}, 66\penalty0 (336):\penalty0 846--850, 1971.

\bibitem[Rubin(1976)]{Rubin1976}
D.~B. Rubin.
\newblock Inference and missing data.
\newblock \emph{Biometrika}, 63\penalty0 (3):\penalty0 581--592, 1976.

\bibitem[Snijders(2011)]{snijders2011statistical}
T.~A. Snijders.
\newblock Statistical models for social networks.
\newblock \emph{Annual Review of Sociology}, 37:\penalty0 131--153, 2011.

\bibitem[Snijders and Nowicki(1997)]{snijders1997estimation}
T.~A. Snijders and K.~Nowicki.
\newblock Estimation and prediction for stochastic blockmodels for graphs with
  latent block structure.
\newblock \emph{J. class.}, 14\penalty0 (1):\penalty0 75--100, 1997.

\bibitem[Szklarczyk et~al.(2015)Szklarczyk, Franceschini, Wyder, Forslund,
  Heller, Huerta-Cepas, Simonovic, Roth, Santos, Tsafou, et~al.]{string:2015}
D.~Szklarczyk, A.~Franceschini, S.~Wyder, K.~Forslund, D.~Heller,
  J.~Huerta-Cepas, M.~Simonovic, A.~Roth, A.~Santos, K.~P. Tsafou, et~al.
\newblock String v10: protein--protein interaction networks, integrated over
  the tree of life.
\newblock \emph{Nucleic. Acids Res.}, 43, 2015.

\bibitem[Thompson and Frank(2000)]{thompson2000model}
S.~K. Thompson and O.~Frank.
\newblock Model-based estimation with link-tracing sampling designs.
\newblock \emph{Survey Methodology}, 26\penalty0 (1):\penalty0 87--98, 2000.

\bibitem[Thompson and Seber(1996)]{thompsonseber}
S.~K. Thompson and G.~Seber.
\newblock \emph{Adaptive Sampling}.
\newblock New-York : Wiley, 1996.

\bibitem[Vinayak et~al.(2014)Vinayak, Oymak, and Hassibi]{Vinayak2014}
R.~K. Vinayak, S.~Oymak, and B.~Hassibi.
\newblock Graph clustering with missing data: Convex algorithms and analysis.
\newblock \emph{Adv. Neu. In.}, 2014.

\bibitem[Vincent and Thompson(2015)]{VincentKyleandThompson2015}
K.~Vincent and S.~Thompson.
\newblock {Estimating the size and distribution of networked populations with
  snowball sampling}.
\newblock Technical report, 2015.
\newblock URL \url{http://arxiv.org/abs/1402.4372v2}.

\end{thebibliography}



\end{document}